\documentclass[runningheads]{llncs}
\usepackage[T1]{fontenc}
\usepackage{amsmath}
\usepackage{multicol}
\usepackage{amsfonts}
\usepackage{xcolor}
\usepackage{graphicx}
\usepackage{tcolorbox}

\usepackage{tcolorbox}
\usepackage{tikz}
\usepackage[linesnumbered, ruled, vlined]{algorithm2e}
\usepackage{caption}
\usepackage{subcaption}
\usepackage{cite}

\usepackage{url}
\usepackage{hyperref}
\hypersetup{
    colorlinks=true,
    linkcolor=blue,
    filecolor=magenta,      
    urlcolor=cyan,
}
\begin{document}
\title{Multi-Slot Tag Assignment Problem in Billboard Advertisement}
% \title{Effective Tagging for Billboard Ads: Zone-specific Budget Optimization}
\author{Dildar Ali \and Suman Banerjee \and Yamuna Prasad}
\authorrunning{Ali et al.} % abbreviated author list (for running head)
\institute{Indian Institute of Technology Jammu,
J \& K-181221, India. \\
\email{\{2021rcs2009,suman.banerjee,yamuna.prasad\}@iitjammu.ac.in}}
\maketitle
\begin{abstract}
  Nowadays, billboard advertising has emerged as an effective advertising technique due to higher returns on investment. Given a set of selected slots and tags, how to effectively assign the tags to the slots remains an important question. In this paper, we study the problem of assigning tags to the slots such that the number of tags for which influence demand of each zone is satisfied gets maximized. Formally, we call this problem the \textsc{Multi-Slot Tag Assignment} Problem. The input to the problem is a geographical region partitioned into several zones, a set of selected tags and slots, a trajectory, a billboard database, and the influence demand for every tag for each zone. The task here is to find out the assignment of tags to the slots, such the number of tags for which the zonal influence demand is satisfied is maximized. We show that the problem is \textsf{NP-hard}, and we propose an efficient approximation algorithm to solve this problem. A time and space complexity analysis of the proposed methodology has been done. The proposed methodology has been implemented with real-life datasets, and a number of experiments have been carried out to show the effectiveness and efficiency of the proposed approach. The obtained results have been compared with the baseline methods, and we observe that the proposed approach leads to a number of tags whose zonal influence demand is satisfied.

 % Considering the zonal influence constraint, this paper studies effectively maximizing the tag allocation to the billboard slots under budget constraints. We call this problem \textsc{Multi-Slot Tag Assignment in Billboard Advertisement}. The inputs to this problem are the Trajectory and Billboard Database, a zonal division of the underlying space, tags, and the influence demand of each zone. This problem asks for the assignment of tags to the slots such that the influence demand of every zone is satisfied and the aggregated allocation cost over all the tags is minimized. We formulated this problem as an Integer Linear Program and proposed a greedy solution approach with an approximation guarantee. The proposed methodology has been analyzed to understand its time and space requirements and implemented with real-world datasets. The results were compared with the baseline methods, and the proposed tag assignment approach was observed to handle more tags than the baseline methods.

\end{abstract}
\section{Introduction}
Billboard Advertisement has emerged as an effective out-of-home advertisement technique due to its ensured return on investment\footnote{\url{https://topmediadvertising.co.uk/billboard-advertising-statistics/}}. In this advertisement technique, an influence provider (e.g., Lamar, Sigtel, etc.) owns several digital billboards, and they can be allocated slot-wise to a group of advertisers depending on their budget. One of the relevant research questions that have been addressed in the literature \cite{ali2022influential,ali2023influential,ali2024influential,zhang2020towards} is that, given a positive integer $k$, which slots should be chosen such that the influence is maximized? The literature has mentioned that the probability of a user being influenced by advertisement content strongly depends on the context. In the literature \cite{ali2024influential}, the notion of context has been formalized as a tag. The influence probability of a user is dependent on the tag, which implies that it is not only sufficient to select influential slots, but it is equally important to allocate tags to the appropriate slots as well. In the billboard advertisement, the advertiser's goal will be to maximize the allocation of tags to slots under budget constraints, and, at the same time, the zone-specific influence requirement of tags should be satisfied. So, one important question arises: given a set of billboard slots and tags, how do we allocate tags to the slots so that an advertiser's allocated tags are maximized? We call this the \textsc{Multi-Slot Tag Assignment in Billboard Advertisement (MSTA)}.

\paragraph{\textbf{Background.}}
To the best of our knowledge, Zhang et al. \cite{10.1145/3219819.3219946} was the first to introduce the influence maximization problem in billboard advertisements, and later, several studies \cite{ali2022influential,ali2023influential} were studied in a similar direction. Later, Zhang et al. \cite{zhang2020towards} further extended their work to find influential billboards under budget constraints. Similarly, Wang et al. \cite{8604082} find the best $k$ influential billboard considering targeted adverting. There exist a few works of literature \cite{ali2024regret,zhang2021minimizing} on regret minimization in the context of billboard advertisements. Next, Ali et al. \cite{ali2024influential} were the first to introduce influential slots and tags selection problem in the context of billboard advertisement. Few studies \cite{ali2024regret,ali2024minimizing} have also considered zone-specific influence demand and budget jointly in multi-advertiser settings. None of the studies has considered the problem of allocating tags to the billboard slots in the context of the billboard advertisement. Next, we discuss the motivation behind this tag allocation problem under zonal influence constraint. 

\paragraph{\textbf{Motivation.}}
Previous studies \cite{8604082,ali2022influential,ali2023influential,ali2024influential,yang2014modeling,zhang2019optimizing,zhang2020towards,zhang2021minimizing} assumed that all zones are equally valuable to advertisers when selecting billboard slots. However, this assumption doesn't always hold in real-life scenarios \cite{ali2024minimizing,ali2024regret}. Consider a city where the advertiser is aware of the zone-specific populations. Assigning advertising content (i.e., tag) for a high-end product to a zone where economically weaker people live would be ineffective, as they are unlikely to purchase the product, wasting the advertiser's budget. Therefore, the advertiser's objective should be to allocate tags to the most suitable billboard slots belonging to different zones to maximize total influence within budget constraints. So, studying the \textsc{Multi-Slot Tag Assignment in Billboard Advertisement under the Zonal Influence constraint} problem is important.

\paragraph{\textbf{Our Contribution.}}   
This paper introduces the problem of tag allocation in billboard advertisements. The key contributions of our work are as follows:
\begin{itemize}
    \item We formulate this as a discrete optimization problem and show this problem is NP-Hard and hard to approximate within a constant factor.
    \item We propose a Cost-Effective Greedy algorithm to solve this problem.
    \item We analyze the time and space requirements and provide an approximation guarantee for the proposed solution approach.
    \item We have shown a theoretical analysis of the proposed method and experimental results on the real-world datasets. 
\end{itemize}

\paragraph{\textbf{Organization of the Paper.}} The rest of the paper has been organized as follows. Section \ref{Sec:Problem} describes background concepts and defines our problem formally. The proposed solution approach has been described in Section \ref{Sec:PSA}. Section \ref{Sec:Experimental_Evaluation} contains the experimental evaluation of the proposed solution approach. Finally, Section \ref{Sec:Conclusion} concludes our study.

\section{Preliminaries and Problem Definition} \label{Sec:Problem}
This section introduces the preliminary concepts and formally defines our problems. We first describe billboard, billboard slot, and influence in the following.

\paragraph{\textbf{Billboard, Billboard Slot, Influence}} In this section, we describe some preliminary concepts and define our problem formally. Consider a set of $n$ billboards $B^{'}=\{b_1, b_2, \ldots, b_n\}$ are owned by an influence provider and placed at different locations (e.g., street junctions, airports, metro stations, shopping malls, etc.) in a city. We assume that all the billboards are running for $\Delta$ time duration, and the goal is to play some display content on digital billboards so that a bulk of people can observe the content, which will help brand a product. In practice, every billboard slot is associated with some cost, and a commercial house can rent the slot by paying the cost to the influence provider. We formally state the notion of billboard slot in Definition \ref{Def:Billboard-slot}.
\begin{definition}[Billboard Slot] \label{Def:Billboard-slot}
A billboard slot of duration $\Delta$ is defined as a tuple $(b_i,[t,t+\Delta])$ where $b_i \in B^{'}$ and $t \in \{T_1, T_1+\Delta+1, T_1+2\Delta+1, \ldots, T_2 - \Delta+1 \}$.
\end{definition}
The set of all billboard slots is denoted by $\mathcal{BS}$, and each of them is associated with some cost which is formalized as a cost function $C: \mathcal{BS} \longrightarrow \mathbb{R}_{\geq 0}$. For any slot $bs \in \mathcal{BS}$, the cost of $bs$ is denoted as $C(bs)$. It can be observed that if there are $n$ billboards and they are operated for the interval $[T_1, T_2]$, then the number of total slots is $m=\frac{T_2-T_1}{\Delta} \cdot n$. Now, from the available slots, depending on the available budget $\mathcal{B}$, the advertiser needs to choose a subset of the slots. Now, the question arises of how to quantify the influence of a subset of billboard slots. We state this notion in Definition \ref{Def:Influence}.

\begin{definition} [Influence of a Subset of Billboard Slots] \label{Def:Influence}
Given a trajectory database  $\mathbb{T}$, billboard database $\mathbb{B}$, a subset of billboard slots $S \subseteq \mathcal{BS}$, and influence probability from every billboard to trajectory, i.e., $Pr(b_i,t_j)$ for each $b_i \in \mathbb{B}$ and $t_j \in \mathbb{T}$, the aggregated influence of all the slots in $S$ can be given by Equation \ref{Eq:Influence}.
    \begin{equation} \label{Eq:Influence}
    \mathcal{I}(\mathcal{S})= \underset{t_{j} \in \mathbb{T}}{\sum} [1-  \underset{b_{i} \in \mathcal{S}}{\prod}(1-Pr(b_{i}, t_{j}))]
    \end{equation}
\end{definition}

The influence function $\mathcal{I}()$ assigns each subset of billboard slots to its corresponding influence value, defined as $\mathcal{I}: 2^{\mathcal{BS}} \rightarrow \mathbb{R}^{+}_{0}$, with $\mathcal{I}(\emptyset) = 0$. This function, as described in Equation \ref{Eq:Influence}, is commonly used in advertising literature \cite{ali2023influential,zhang2019optimizing,zhang2020towards,zhang2021minimizing}.

\paragraph{\textbf{Zonal Influence.}}
Consider the city under consideration is divided into $\ell$ zones $\mathcal{Z}=\{z_1, z_2, \ldots, z_{\ell}\}$. For any subset of billboard slots $\mathcal{S}$, its influence confined to any zone $z_i \in \mathcal{Z}$ is denoted as $I_{z_i}(\mathcal{S})$. The zonal influence constraint ensures that the influence demand of a tag from every zone is satisfied. Next, we elaborate on the budget and zonal influence constraints.
% \paragraph{\textbf{Problem Definition.}}
% Now, we define our problem formally. We call our problem \textsc{Multi-Slot Tag Assignment in Billboard Advertisement}. The inputs to the problem are a trajectory database $\mathbb{T}$, a billboard database $\mathbb{B}$, the cost function $C$, the zonal influence constraint vector $\mathcal{V}=\{\sigma_1, \sigma_2, \ldots, \sigma_{\ell}\}$, and a budget $\mathcal{B}$. This problem asks for a subset of the billboard slots to allocate tags by satisfying budget and zonal influence constraints. Before that, we elaborate on both the constraints:
\begin{itemize}
    \item \textbf{Budget Constraint}: This constraint emphasizes that the total selection cost should not exceed the budget, i.e., $\underset{bs \in \mathcal{S}}{\sum} C(bs) \leq \mathcal{B}$.
    \item \textbf{Zonal Influence Constraint}: This constraint emphasizes that the influence demand of every zone is satisfied, i.e., for every $i \in \{1,2, \ldots, \ell\}$, $I_{z_i}(\mathcal{S}) \geq \sigma_{i}$, where $\sigma_{i}$ denotes the zonal influence demand of the $i^{th}$ tag.
\end{itemize}
% In this study, we pose this problem as a discrete optimization, which is mentioned in the definition.
% \begin{equation}
%     \mathcal{S} \longleftarrow \underset{\mathcal{S} \subseteq \mathcal{BS}, \underset{b \in \mathcal{S}}{\sum} C(b) \leq B \text{ and }I_{z_i}(\mathcal{S}) \geq \sigma_{i} \text{ for all } i \in \{1,2, \ldots, \ell\}}{argmax} \ I(\mathcal{S})
% \end{equation}

% Here, $\mathcal{S}$ denotes the optimal subset of slots for the budget $\mathcal{B}$.  Next, we model and define three different problems.
\paragraph{\textbf{Model and Problem.}} 
We address the multi-slot tag assignment problem $(\mathcal{A}, \mathcal{S}, $ $ \mathcal{T}, \mathcal{V}, \mathcal{C})$ involving an advertiser $\mathcal{A}$, a set of slots \(\mathcal{S} = \{bs_{1}, bs_{2}, \ldots, bs_{m}\}\), a set of tags \(\mathcal{T} = \{t_{1}, t_{2}, \ldots, t_{k}\}\), $\mathcal{V}=\{\sigma_1, \sigma_2, \ldots, \sigma_{\ell}\}$ denotes the zonal influence demand of tags of advertiser $\mathcal{A}$, and a non-negative cost matrix $\mathcal{C} = \{c_{ij} ~|~ bs_{i} \in \mathcal{S}, t_{j} \in \mathcal{T}\}$ which denotes the cost $c_{ij}$ for allocating tag $t_{j}$ to slot $bs_{i}$. The advertiser $\mathcal{A}$ is associated with $k$ number of tags, which can only be allocated to slots within specific zones where they are demanded.  At most, one tag can be allocated to one slot, as in the OMBM setting (e.g., see  \cite{ali2024effective}). Let a vector \(\mathcal{P} = \{p_{1}, p_{2}, \ldots, p_{k}\}\), where \(p_{i}\) contains the slots in which tag \(t_{i}\) is assigned and $\mathcal{I}(p_{i})$ denotes influence generated by slots in $p_{i}$. A tag $t_{i}$ is considered handled if the zone-specific influence demand $t_{i}(\sigma_{j})$, where $j \in \{1,2,\ldots,\ell\}$ is met; otherwise, it is not handled. In this paper, we focus on a major problem in which the budget for advertising tags is given \(\mathcal{B} \in \mathbb{Z}^{+}\), and subject to the zonal influence constraint, our main goal is to maximize the number of handled tags of advertiser $\mathcal{A}$. Next, we define our problems from the computational perspective. 
\par Now, given total budget \(\mathcal{B}\), we define the cost for handling advertiser $\mathcal{A}$ as the sum of the costs of the slots allocated to satisfy all the tags zone-specific influence demand \(\sigma\) associated with $\mathcal{A}$. We aim to maximize the number of handled tags while ensuring that the zonal influence constraint is met and the total cost of allocated slots does not exceed \(\mathcal{B}\).
\begin{center}
\begin{tcolorbox}
\underline{\textsc{Multi-Slot Tag Assignment in Billboard Advertisement (MSTA)}} \\
\textbf{Input:} $(\mathcal{A}, \mathcal{S}, \mathcal{T}, \mathcal{V}, \mathcal{C})$ and $\mathcal{B}$

\textbf{Task:} Find out an allocation such that the maximum number of tags of an advertiser can be handled subject to the zonal influence constraint and having total cost at most $\mathcal{B}$.
\end{tcolorbox} 
\end{center}
% Next, we formalize our problems using integer linear programming. Assume, $x_{ij} \in \{0, 1\}$ denotes whether slot $s_{i}$ is assigned to advertiser $t_{j}$ and $y_{j} \in \{ 0, 1\}$ denotes is handled. Now, the \textsc{MaxTagUnderTotalBudget} can be formulated in integer linear programming as $\mathcal{IP}$.

Next, we formalize our problems using integer linear programming. Assume, $x_{ij} \in \{0, 1\}$ denotes whether slot $bs_{i}$ is assigned to tag $t_{j}$ of advertiser $\mathcal{A}$ and $y_{j} \in \{ 0, 1\}$ denotes tag $t_{j}$ is handled. This problem can be formulated in integer linear programming as $\mathcal{IP}$.

% \begin{center}
% \begin{tcolorbox}
% \underline{\textsc{Formulated Integer Program $(\mathcal{IP})$}} \\
% \begin{tabular}{l | r}
% \parbox{0.45\textwidth}{
% \textbf{Maximize} $\quad \sum_{j=1}^{m} y_{j}$ \\
% \textbf{subject to} \\
% (1) $\quad \sum_{i=1}^{n} \sum_{j=1}^{m} c_{ij} x_{ij} \leq \mathcal{B},$ \\
% (2) $\quad \sum_{i=1}^{m} \mathcal{I}(p_{i}) \geq \sum_{i=1}^{m} \sum_{j=1}^{\ell} t_{i}(\sigma_{j})$, $\quad \forall ~ p \in \mathcal{P}$\\
% }
% &
% \parbox{0.6\textwidth}{
% (3) $\quad \sum_{j=1}^{m} x_{ij} \leq 1, \quad \forall ~i \in [n]$,\\
% (4)$\quad x_{ij} \in \{0,1\}, \quad \forall ~i \in [n], \forall ~j \in [m],$ \\
% (5)$\quad y_{j} \in \{0,1\}, \quad \forall ~j \in [m]$
% }
% \end{tabular}
% \end{tcolorbox}
% \end{center}

\begin{center}
\begin{tcolorbox}
\underline{\textsc{Formulated Integer Program $(\mathcal{IP})$}} \\
\begin{tabular}{l | r}
\parbox{0.45\textwidth}{
\textbf{Maximize} $\quad \sum_{j=1}^{k} y_{j}$ \\
\textbf{subject to} \\
(1) $\quad \sum_{i=1}^{m} \sum_{j=1}^{k} c_{ij} x_{ij} \leq \mathcal{B},$ \\
(2) $\quad \sum_{i=1}^{k} \mathcal{I}(p_{i}) \geq \sum_{i=1}^{k} \sum_{j=1}^{\ell} t_{i}(\sigma_{j})$, $\quad \forall ~ p \in \mathcal{P}$\\
}
&
\parbox{0.6\textwidth}{
(3) $\quad \sum_{j=1}^{k} x_{ij} \leq 1, \quad \forall ~i \in [m]$,\\
(4)$\quad x_{ij} \in \{0,1\}, \quad \forall ~i \in [m], \forall ~j \in [k],$ \\
(5)$\quad y_{j} \in \{0,1\}, \quad \forall ~j \in [k]$
}
\end{tabular}
\end{tcolorbox}
\end{center}

The first constraint in the above $\mathcal{IP}$ ensures that the total allocation cost does not exceed the budget $\mathcal{B}$. The second constraint guarantees that every handled tag's zone-specific influence demand should be satisfied. The third constraint ensures that, at most, one tag can be allocated to a slot. 

\section{Proposed Solution Approach}\label{Sec:PSA}
In this section, we consider an approximation algorithm for our problem and demonstrate that the proposed greedy algorithm achieves a tight approximation of \( \max_{i}~ q_{i} + 1 \) for the \textsc{MSTA} problem. This type of greedy approach has been previously studied by Aziz et al. \cite{10.5555/3463952.3463974} for multi-robot task allocation problems, and we extend this approach to our problem context. 

\paragraph{\textbf{Cost-Effective Greedy (CEG) Algorithm.}}
First, we initialize \(\mathcal{R}\) and sets \(\mathcal{Y}\) in Line No. $1$ and $2$ and in Line No. $3$ \texttt{while loop} will execute till $\mathcal{T} \neq \emptyset$. Next, in the \texttt{for} loop at lines 4 and 6, for each tag and within each tag, for each zone, we allocate tags to slots such that minimal cost is incurred. Initially, in lines 7 to 9, a subset of slots is selected for each zone according to the zone-specific influence demand associated with the advertiser's tag. In line 10, the allocation cost for the selected slots and corresponding tags is calculated. Next, in line 12, the best allocation with the minimum cost is chosen, and in line 13, we check if adding the new allocation with minimum cost, \(\mathcal{C}(t^{*}_{j}, \mathcal{S}^{*}_{j})\), does not exceed the total budget \(\mathcal{B}\). If this condition is satisfied, the allocation is added to the allocation vector \(\mathcal{Y}\), and the corresponding slots and tags are removed (lines 13 to 17). Otherwise, the \texttt{if} condition is not satisfied, or \(\mathcal{T}\) is empty, causing the \texttt{while} loop at line 3 to break. The proposed methodology has been described in the form of pseudocode in Algorithm \ref{Algo:1}.

\begin{algorithm}[h]
\scriptsize
\SetAlgoLined
\KwData{$(\mathcal{A}, \mathcal{S}, \mathcal{T}, \mathcal{V}, \mathcal{C}, \mathcal{B})$.}
\KwResult{An allocation $\mathcal{Y} = \{(t_{1}, s_{1}), \ldots, (t_{k}, s_{k})\}$ of tags to slots}
$\mathcal{R} \leftarrow \emptyset$\; 
 Initialize $\mathcal{Y} = \{\mathcal{Y}_{1}, \mathcal{Y}_{2}, \ldots \mathcal{Y}_{k}\}$\;
\While{$\mathcal{T} \neq \emptyset$}{
\For{$t_{j} \in \mathcal{T}$}{
$\mathcal{R}^{'} \leftarrow \emptyset$\;
\For {each $z_{\ell} \in a_{z_j}$}{
$\mathcal{R}=
\begin{cases}
\underset{s^{'} \subseteq \mathcal{S}_{z_\ell} : \mathcal{I}(s^{'}) \geq \sigma_{z_\ell}}{argmin} ~ ~\underset{s_{i} \in s^{'}}{\sum} \mathcal{C}_{ij}, & \text{if } \mathcal{I}(\mathcal{S}_{z_\ell}) \geq \sigma_{\ell} \\
\emptyset, & \text{otherwise}
\end{cases}$ \\
$\mathcal{R}^{'} \leftarrow \mathcal{R}^{'} \cup \mathcal{R}$\;
}
$\text{Set}~ \mathcal{S}_{j} \leftarrow \mathcal{S}_{j} \cup \mathcal{R}^{'}$\;
$\text{Set} ~ \mathcal{C}(t_{j}, \mathcal{S}_{j})=
\begin{cases}
\underset{s_{i} \in \mathcal{S}_{j}}{\sum} \mathcal{C}_{ij}, & \text{if } \mathcal{S}_{j} \neq \emptyset \\
+ \infty, & \text{otherwise}
\end{cases}$ \\
}
$\text{Set} ~ \mathcal{C}(t^{*}_{j}, \mathcal{S}^{*}_{j})= \underset{t_{j} \in \mathcal{T}}{argmin} \ \mathcal{C}(t_{j}, \mathcal{S}_{j})$\;
\If {$\underset{(t_{j}, \mathcal{S}_{j}) \in \mathcal{Y}}{\sum}~ \underset{s_{i} \in \mathcal{S}_{j}}{\sum} \mathcal{C}_{ij} + \mathcal{C}(t^{*}_{j}, \mathcal{S}^{*}_{j}) \leq \mathcal{B}$}{
$\mathcal{Y} \leftarrow \mathcal{Y} \cup \{ (t^{*}_{j}, \mathcal{S}^{*}_{j})\}$\;
$\mathcal{T} \leftarrow \mathcal{T} \setminus \{t^{*}_{j}\}$\;
$\mathcal{S} \leftarrow \mathcal{S} \setminus \{s^{*}_{j}\}$\;
}
\Else {break}
}
return $\mathcal{Y}$ \;
\caption{Cost-Effective Greedy (CEG) Algorithm for Multi-Slot Tag Assignment Problem}
\label{Algo:1}
\end{algorithm}
\vspace{-0.1in}
\paragraph{\textbf{Complexity Analysis.}}
Now, we analyze Algorithm \ref{Algo:1} for the time and space requirements. In-Line No. $1$ and $2$ initializing $\mathcal{R}$ and $\mathcal{Y}$ will take $\mathcal{O}(1)$ and $\mathcal{O}(k)$, respectively. In-Line No. $4$ \texttt{while loop} will run for $\mathcal{O}(k)$ times and in Line No. $4$ \texttt{for loop} will execute for $\mathcal{O}(k^{2})$ times. Next, in Line No., calculating cost will take $\mathcal{O}(m \cdot k)$ and in Line No. $6$ to $9$ \texttt{for loop} will take $\mathcal{O}(k^{2} \cdot \ell \cdot m \cdot t + m \cdot k^{2})$ where $\ell$, $m$, and $t$ are the number of demographic zones, number of billboard slots and number of tuples in the trajectory database to calculate influence value, respectively. Similarly, Line No. $10$ and $11$ will take $\mathcal{O}(k^{2})$ and $\mathcal{O}(k^{2} \cdot m)$ times. So, Line No. $4$ to $12$ will take $\mathcal{O}(k^{2} \cdot \ell \cdot m \cdot t + m \cdot k^{2} + k^{2} + k^{2} \cdot m)$ time to execute. In-Line No. $13$ will execute for at most $\mathcal{O}(k^{2})$ time and Line No. $14$ will execute till the budget constraint $\mathcal{B}$ and $\mathcal{T} \neq \emptyset$ satisfies. So, Line No. $14$ to $18$ will execute for $\mathcal{O}(k \cdot m)$ time in the worst case. Therefore Algorithm \ref{Algo:1} will take total $\mathcal{O}(k^{2} \cdot \ell \cdot m \cdot t + m \cdot k^{2} + k^{2} + k^{2} \cdot m + k \cdot m)$ i.e., $\mathcal{O}(k^{2} \cdot \ell \cdot m \cdot t)$. The additional space requirement for storing $\mathcal{Y}, \mathcal{R}$, will be $\mathcal{O}(k)$ and $\mathcal{O}(1)$, respectively.
\begin{theorem}\label{Th:Complexity}
Time and space requirement by Algorithm \ref{Algo:1} is of $\mathcal{O}(k^{2} \cdot \ell \cdot m \cdot t)$ and $\mathcal{O}(k)$, respectively.
\end{theorem}

\paragraph{\textbf{Description on Performance Guarantee.}}
Now, we proceed to describe the performance guarantee of the proposed solution approach. In particular, we prove Theorem \ref{TH2}.

\begin{theorem}\label{TH2}
The greedy solution in Algorithm \ref{Algo:1} for the \textsc{MSTA} problem can be implemented in polynomial time, and it can achieve the approximation of $(\mathcal{P}^{*} +1)$ where $\mathcal{P}^{*} = \underset{j \in \mathcal{T}}{Max~~ \mathcal{P}_{j}}$.
\end{theorem}
\begin{proof}
First, we show that the proposed greedy approach in Algorithm \ref{Algo:1} can be implemented in polynomial time. The \texttt{while loop} at Line No. $3$ can be repeated at most $m$ times, and in each iteration, the value of $|\mathcal{T}|$ can be decreased by $1$. In each round of \texttt{while loop} for each tag, two \texttt{for loop} can be run for the polynomial time. This happens because to find the minimum cost for allocating each tag to slots only needs the minimum $p_{i}$ costs between tags and all the slots. Next, we prove the approximation ratio and assume the solution of the greedy algorithm is $\mathcal{Y} = \{(t_{1}, \mathcal{S}_{1}), \ldots, (t_{k}, \mathcal{S}_{k})\}$, where $k$ is the number of tags that are handled. For each $(t_{j}, \mathcal{S}_{j}) \in \mathcal{T} \times 2^{\mathcal{S}}$ denotes tag $t_{j}$ is handled by the slots in $\mathcal{S}$. Let there be an optimal solution $\mathcal{Q} =  \{(t_{j1}, \mathcal{S}^{'}_{j1}), \ldots, (t_{j\ell}, \mathcal{S}^{'}_{j\ell})\}$ exists. Now, we have to prove that $k \geq \frac{\ell}{\mathcal{P}^{*} + 1}$.

\par To find out the difference between $\mathcal{Y}$ and $\mathcal{Q}$, we observe that for each $1 \leq j \leq k$, $|\mathcal{S}_{j}| = p_{j} \leq \mathcal{P}^{*}$ and tag $t_{j}$ can be  impede at most $\mathcal{P}^{*} + 1$ tags in $\mathcal{Q}$ including $t_{j}$ itself. Now, for each $t_{j} \in \mathcal{Y}$, we recursively define $\mathcal{Q}_{j} \subseteq \mathcal{Q} \setminus \underset{e \leq j-1}{\bigcup ~ \mathcal{Q}_{e}}$ as the tags are not handled because of $t_{j}$ i.e., $\mathcal{Q} = \{ (t_{j}, \mathcal{S}^{'}_{j}) \in \mathcal{Q} \setminus \underset{e \leq j-1}{\bigcup ~ \mathcal{Q}_{e}}~ ~ |~ \mathcal{S}^{'}_{j} \cap \mathcal{S}_{j} \neq  \emptyset\}$. It is also clear that in $\mathcal{Q}_{j}$ cannot be impeded by $\{t_{1}, t_{2}, \ldots, t_{j-1}\}$ and moreover, $|\mathcal{Q}_{j}| \leq \mathcal{P}^{*} +1$. We further partition our greedy solution into $\mathcal{Y}_{1}$ and $\mathcal{Y}_{2}$, where $\mathcal{Y}_{1}$ is the set of tags whose selection impedes tags in $\mathcal{Q}$ and it denotes $\mathcal{Y}_{1} = \{(t_{j}, \mathcal{S}_{j}) \in \mathcal{Y}~ ~ | ~ ~ |\mathcal{Q}_{j}| \geq 1\}$. Similar way $\mathcal{Y}_{2} = \mathcal{Y} \setminus \mathcal{Y}_{1}$ and we can claim that, $(\mathcal{P}^{*} +1) \cdot |\mathcal{Y}_{1}| \geq |\underset{(t_{j}, \mathcal{S}_{j}) \in \mathcal{Y}_{1}}{\bigcup ~ ~ \mathcal{Q}_{j}}|$ and this is holds because $|\mathcal{Q}_{j}| \leq \mathcal{P}^{*} + 1$ for any $j$. Accordingly, $\underset{y \in \mathcal{S}_{j}}{\sum ~ ~ \mathcal{C}_{y_j}} \leq \underset{y \in \mathcal{S}^{'}_{j^{'}}}{\sum ~ ~ \mathcal{C}_{y_j^{'}}}$ for each $(t_{j}, \mathcal{S}_{j}) \in \mathcal{Y}_{1}$ and any $(t_{j^{'}}, \mathcal{S}^{'}_{j^{'}}) \in \mathcal{Q}_{j}$ and this inequality holds because Algorithm \ref{Algo:1} always selects a feasible tag slots pair considering minimum allocation costs. There will be contradictions if $\underset{y \in \mathcal{S}_{j}}{\sum ~ ~ \mathcal{C}_{y_j}} > \underset{y \in \mathcal{S}^{'}_{j^{'}}}{\sum ~ ~ \mathcal{C}_{y_j^{'}}}$, since $(t_{j^{'}}, \mathcal{S}^{'}_{j^{'}})$ is not impeded by $\{t_{1}, t_{2}, \ldots t_{j-1}\}$ then $(t_{j^{'}}, \mathcal{S}^{'}_{j^{'}})$ should be added to $\mathcal{Y}$ before $(t_{j}, \mathcal{S}_{j})$. Now, assume $\mathcal{Q}^{'} = \mathcal{Q} \setminus \underset{(t_{j}, \mathcal{S}^{'}_{j}) \in \mathcal{Y}_{1}}{\mathcal{Q}_{j}}$ and we have to show that $|\mathcal{Y}_{2}| \geq |\mathcal{Q}^{'}|$. Next, using the definition of $\mathcal{Y}_{1}$ we can write,

\begin{equation}
    \underset{(t_{j}, \mathcal{S}_{j}) \in \mathcal{Y}_{1}}{\sum} ~\underset{ y \in \mathcal{S}_{j}}{\sum} \mathcal{C}_{y_j} \leq  \underset{(t_{j}, \mathcal{S}_{j}) \in \mathcal{Y}_{1}}{\sum} ~ ~\underset{(j^{'}, \mathcal{S}^{'}_{j^{'}}) \in \mathcal{Q}_{j}}{\sum} ~ ~\underset{ y \in \mathcal{S}^{'}_{j^{'}}}{\sum} \mathcal{C}_{y_{j^{'}}}
\end{equation}

which means the budget assigned for $\mathcal{Y}_{2}$ is sufficient for $\mathcal{Q}^{'}$. Thus we can say $|\mathcal{Y}_{2}| \geq |\mathcal{Q}^{'}|$ and we can write, 

\begin{equation}
    k = |\mathcal{Y}_{1}| + |\mathcal{Y}_{2}| \geq \frac{|\underset{(t_{j}, \mathcal{S}_{j}) \in \mathcal{Y}_{1}}{\bigcup ~ ~ \mathcal{Q}_{j}}|}{(\mathcal{P}^{*} +1)} + |\mathcal{Q}^{'}| \geq \frac{\ell}{\mathcal{P}^{*} + 1}
\end{equation}
which completes the proof of Theorem \ref{TH2}.
\end{proof}

\section{Experimental Evaluation}\label{Sec:Experimental_Evaluation}

\paragraph{\textbf{Dataset Description}}
We choose two different real-life datasets, New York City (NYC)\footnote{\url{https://www.nyc.gov/site/tlc/about/tlc-trip-record-data.page}} and Los Angeles (LA)\footnote{\url{https://github.com/Ibtihal-Alablani}}. These datasets are also used in the existing literature \cite{ali2022influential,ali2023influential,ali2024influential,ali2024regret}. The NYC dataset includes 227,428 check-ins from April 12, 2012, to February 16, 2013, while the LA dataset comprises 74,170 check-ins from 15 different streets in Los Angeles. Additionally, we collected billboard data from various locations in NYC and LA through LAMAR\footnote{\url{http://www..lamar.com/InventoryBrowser}}, a major billboard provider. The NYC billboard dataset contains 1,031,040 slots, and the LA dataset contains 2,135,520 slots. Further, we divide NYC billboard datasets into five geographic regions based on latitude and longitude: Bronx, Brooklyn, Manhattan, Queens, Staten Island, and LA, with three zones based on 15 streets.

\paragraph{\textbf{Billboard Cost.}} Outdoor advertising companies, like LAMAR, typically do not disclose the exact cost of renting a billboard slot. As reported in existing studies \cite{ali2022influential,ali2024effective,zhang2019optimizing,zhang2020towards}, the cost of a slot is assumed to be proportional to its influence. So, we adopt a similar approach, defining the cost of a billboard slot \(bs\) as $\text{Cost}(bs) = \left\lfloor \beta \times \frac{\mathcal{I}(bs)}{10} \right\rfloor$. Here, \(\beta\) is a factor randomly chosen between 0.8 and 1.1 to account for cost fluctuations, and \(\mathcal{I}(bs)\) represents the influence of slot $bs$.

\paragraph{\textbf{Advertiser.}} Given an advertiser $\mathcal{A}$, with a set of tags $\mathcal{T} = \{t_1, t_2, \ldots, t_k\}$, $\mathcal{A}$ submits a campaign proposal to an influence provider to deploy tags the billboard slots under a budget $\mathcal{B}$. The advertiser can be represented as $(\mathcal{A}, t_{i}, \sigma_{i}, \mathcal{B})$, where $t_{i}$ is the $i^{th}$ tag and $\sigma_{i}$ is the influence demand of the corresponding tag.

\subsection{Experimental Setup}
We have considered the following experimental setup for our experiments.
\paragraph{\textbf{Key Parameters.}}
All the parameters are summarized in Table \ref{Key-parameters}. We vary the number of tags from $5$ to $100$. We vary $\lambda$ from $25m$ to $150m$, and $\lambda$ denotes the distance in which billboard slots can influence trajectories. We vary one parameter in each experiment and set the remaining parameters in their default settings (highlighted in bold).
% \vspace{-0.25in}
\begin{table}[h!]
\caption{\label{Key-parameters} Key Parameters}
\vspace{-0.15 in}
\begin{center}
    \begin{tabular}{ | p{2cm}| p{5.5cm}|}
    \hline
    Parameter & Values  \\ \hline
    $\theta$ & $40\%, 60\%, 80\%, \textbf{100\%}, 120\%$   \\ \hline
    $\delta$ & $1\%, 2\%, \textbf{5\%}, 10\%, 20\%$   \\ \hline
    $\mathcal{|T|}$ & $5, 10, 20, 50, 100$   \\ \hline
    $\lambda$ & $25m,50m,\textbf{100m},125m,150m$  \\ \hline
    \end{tabular}
\end{center}
\end{table}
% \vspace{-0.25in}
\paragraph{\textbf{Demand Supply Ratio $\theta$.}}
It is the proportion of the global influence demand over the influence provider influence supply and it can be defined as $\theta = \sigma^{\mathcal{T}} / \sigma^{*}$, where $\sigma^{\mathcal{T}} = \sum_{i=1}^{k} \sigma_{i}$ represents global demand and $\sigma^{*} = \sum_{b \in \mathcal{BS}} \mathcal{I}(b)$ is the influence provider supply. We simulate five different values of $\theta$ as $40\%,$ $60\%,$ $80\%,$ $100\%,$ $120\%$. 

\paragraph{\textbf{Average Individual Demand Ratio $\delta$.}} It is the ratio of average individual demand over the influence provider influence supply, i.e., $\delta = \sigma^{\mathcal{T}^{''}} / \sigma^{*}$, where $\sigma^{\mathcal{T}^{''}} = \sigma^{\mathcal{T}} / |\mathcal{T}|$ is the average individual influence demand of all the tags. This $\delta$ value adjusted the individual tag's demand.

\paragraph{\textbf{Tags Demand $\sigma$.}} We generate the demand of each tags based on $\sigma = \lfloor \omega \cdot \sigma^{*} \cdot \delta \rfloor$, where $\omega$ is a factor randomly chosen between $0.8$ to $1.2$ to simulate various payment of tag.

\paragraph{\textbf{Advertiser Budget} $\mathcal{B}$.}
We follow a widely adopted payment setting in the existing studies \cite{aslay2017revenue,aslay2015viral,zhang2021minimizing,ali2024minimizing,ali2024regret} and set each tag payment as proportional to its influence demand, i.e., $\mathcal{L}_{i} = \lfloor \alpha \cdot \sigma_{i} \rfloor$, where $\alpha$ is a factor randomly chosen from $0.9$ to $1.1$ to simulate different payment. So, the budget of the advertiser can be defined as $\mathcal{B} = \sum_{i=1}^{k} \mathcal{L}_{i}$.
\paragraph{\textbf{Environment Setup.}}
All the Python codes are executed in HP Z4 workstations with 64 GB RAM and an Xeon(R) 3.50 GHz processor.  
% \paragraph{\textbf{Performance Metrics}}

\subsection{\textbf{Baseline Methods}}
\paragraph{\textbf{Random Allocation.}}
In this approach, billboard slots are selected uniformly at random to allocate tags to them till the zone-specific influence requirement and the budget constraint are satisfied.
\paragraph{\textbf{Top-$k$ Allocation.}}
In this approach, billboard slots are sorted based on their influence value. Next, sorted slots are selected to allocate tags individually until the zonal influence requirements and budget constraints are satisfied.

\subsection{\textbf{Goals of our Experiments}} \label{Sec:Research_Questions}
In this study, we want to address the following Research Questions (RQ).
\begin{itemize}
\item \textbf{RQ1}: Varying $\theta$, $\delta$, how do the number of handled tags change?
\item \textbf{RQ2}: Varying $\theta$, $\delta$, how do the total budget is utilized?
\item \textbf{RQ3}: Varying $\theta$, $\delta$, how do the computational time change?
\item \textbf{RQ4}: Varying $\lambda$, how do the influence quality change?
\end{itemize}

\subsection{\textbf{Observations with Explanation}}
We first evaluate how varying $\theta$ and $\delta$ impact the allocation of tags to billboard slots. As stated previously, we vary $\theta$ from $40\%$ to $120\%$ and $\delta$ from $1\%$ to $20\%$ to answer the research questions (RQ) that arise. In the following section, we will present the results of the LA and NYC datasets under four different cases to show the effectiveness of the study. 
\paragraph{\textbf{ Varying $\theta$, $\delta$ Vs. Tags.}}
Corresponding to case $1$, we have $\theta \leq 80\%$ and $\delta \leq 2\%$, and this refers to both global and individual demand of all the tags being low, i.e., influence provider has a large number of tags with small individual demand. As $\theta$ is small, most of the tag's zone-specific demand is satisfied, and tags are allotted to the slots, shown in Figure \ref{Fig:ADV} (a,b) of LA and Figure \ref{Fig:ADV} (e,f) of NYC. In case $2$, we have $\theta \leq 80\%$ and $\delta \geq 5\%$, and it refers to global demand being lower than the influence supply and the influence provider having tags whose individual influence demand is higher. The influence provider has to deploy more slots to allocate tags, and the number of satisfied tags compared to the case $1$ will be less, as shown in Figure \ref{Fig:ADV} (c,d) of LA and Figure \ref{Fig:ADV} (g,h) of NYC. In case $3$, we have $\theta \geq 100\%$ and $\delta \leq 2\%$ represents a situation where global demand is very high and individual demand is low. The influence provider has a large tags base formed by small tags with low influence demand. For the given $\theta = 100\%$, none of the algorithms can satisfy all the tags shown in Figure \ref{Fig:ADV} (a,b) of LA and Figure \ref{Fig:ADV} (e,f) of NYC. In case $4$, we have $\theta \geq 100\%$ and $\delta \geq 5\%$ represents a situation where none of the algorithms can satisfy all the tags as both individual and global influence demand are high and the number of satisfied tags will be very less as shown in Figure \ref{Fig:ADV} (c,d) of LA and Figure \ref{Fig:ADV} (g,h) of NYC. One point to be highlighted in the NYC datasets is that, in most cases, fewer tags are handled due to unsatisfied zonal influence demand.
% NO OF ADVERTISER PLOT

\begin{figure*}[!ht]
\centering
\begin{tabular}{cccc}
\includegraphics[scale=0.125]{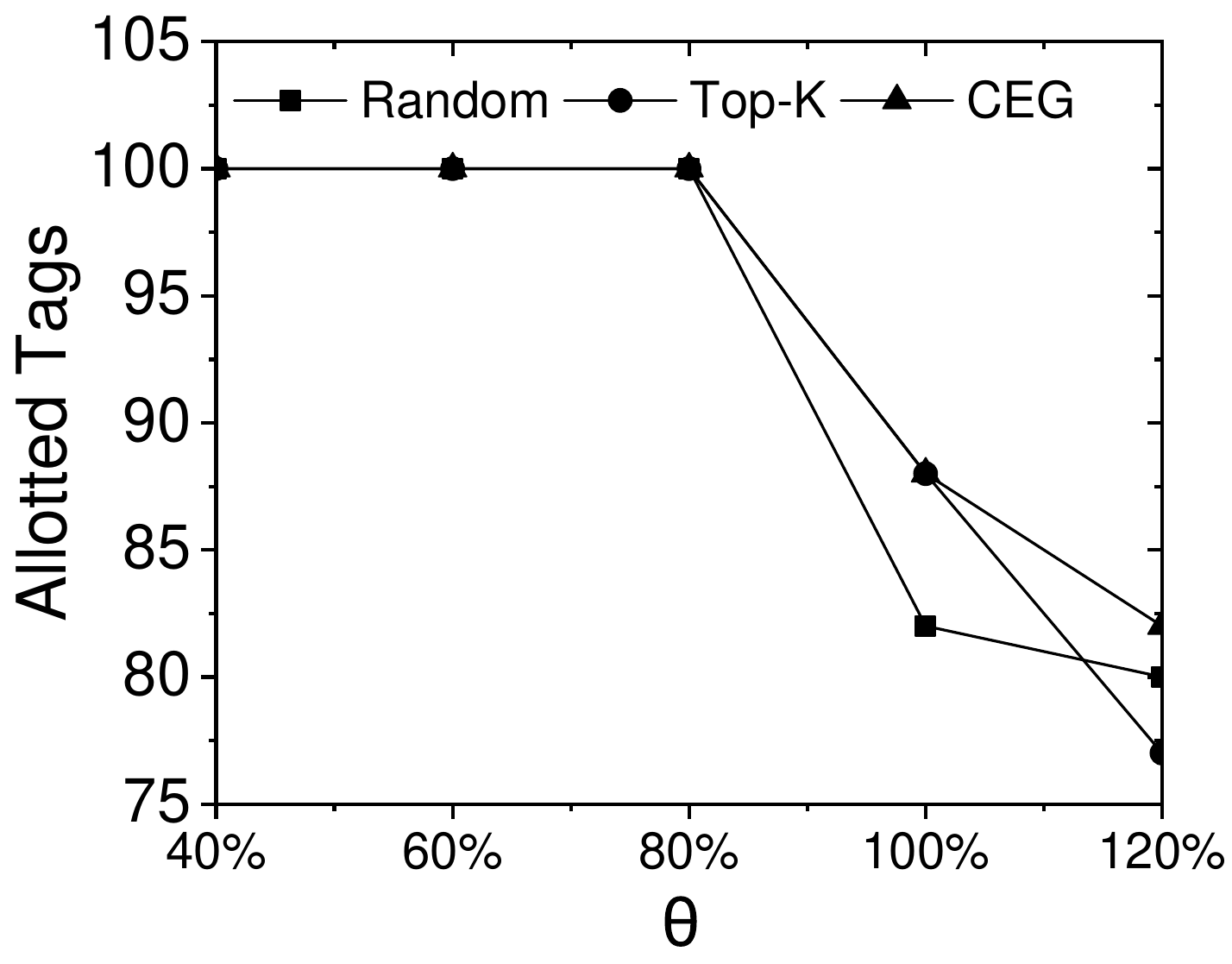} & \includegraphics[scale=0.125]{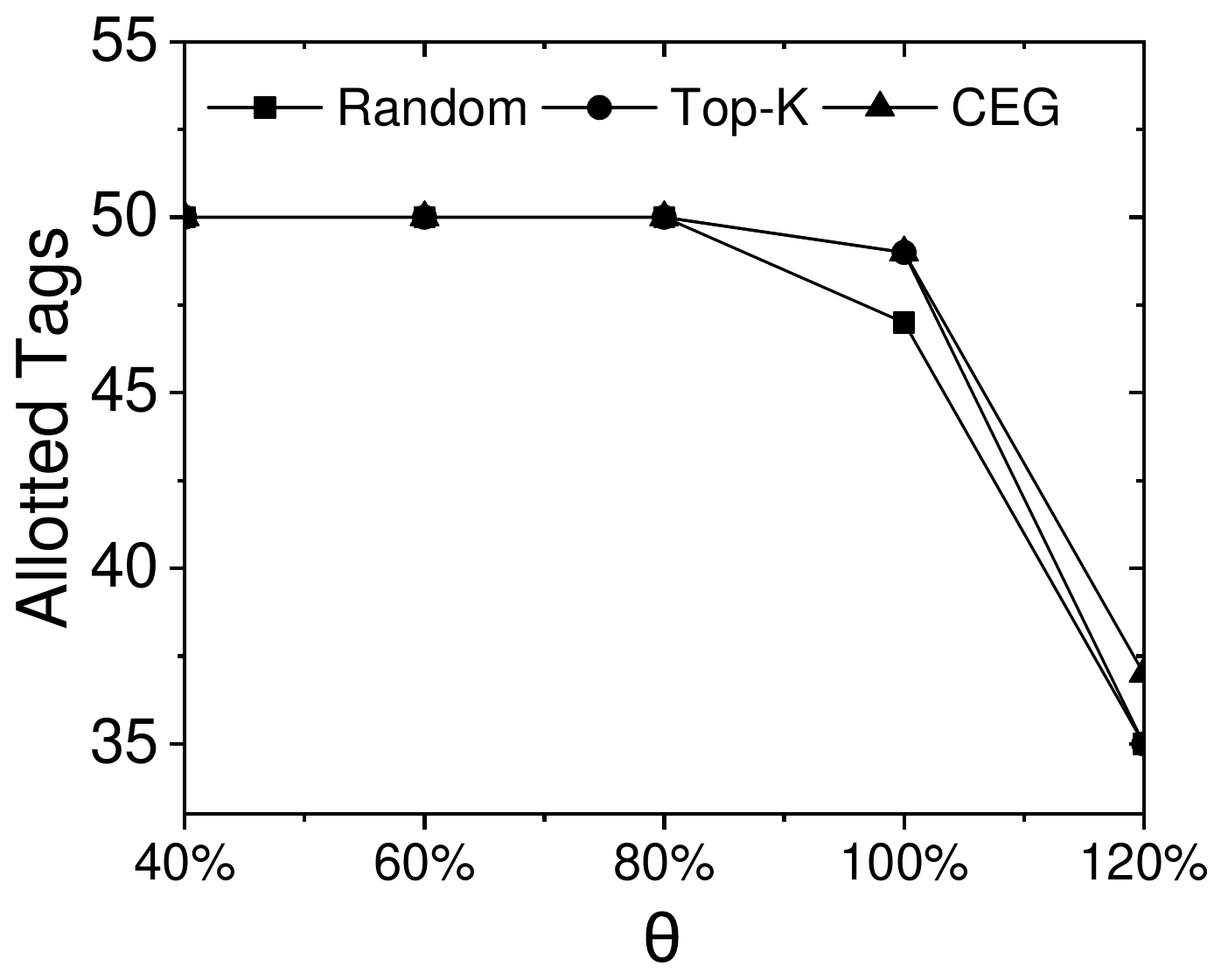} & \includegraphics[scale=0.125]{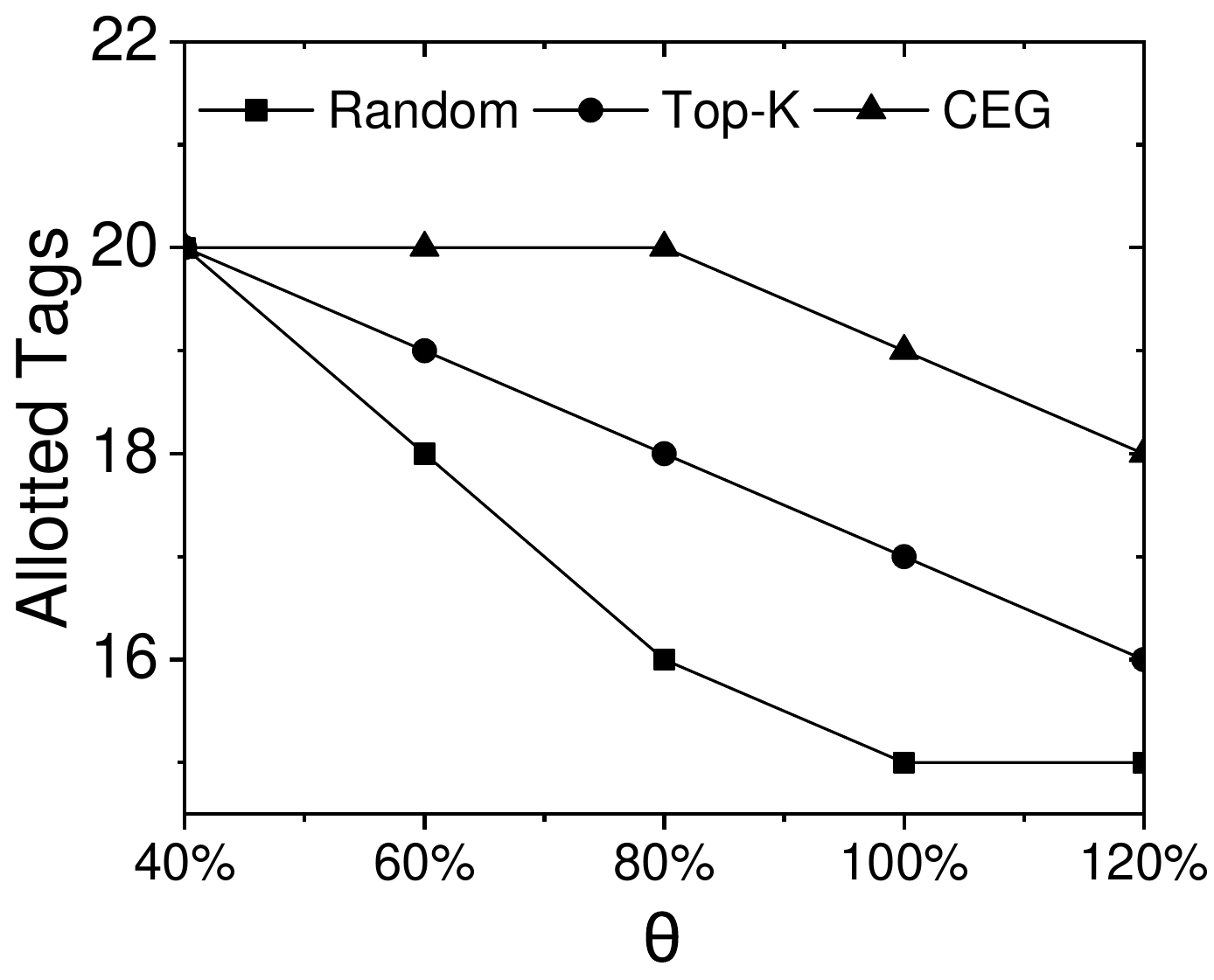} & \includegraphics[scale=0.125]{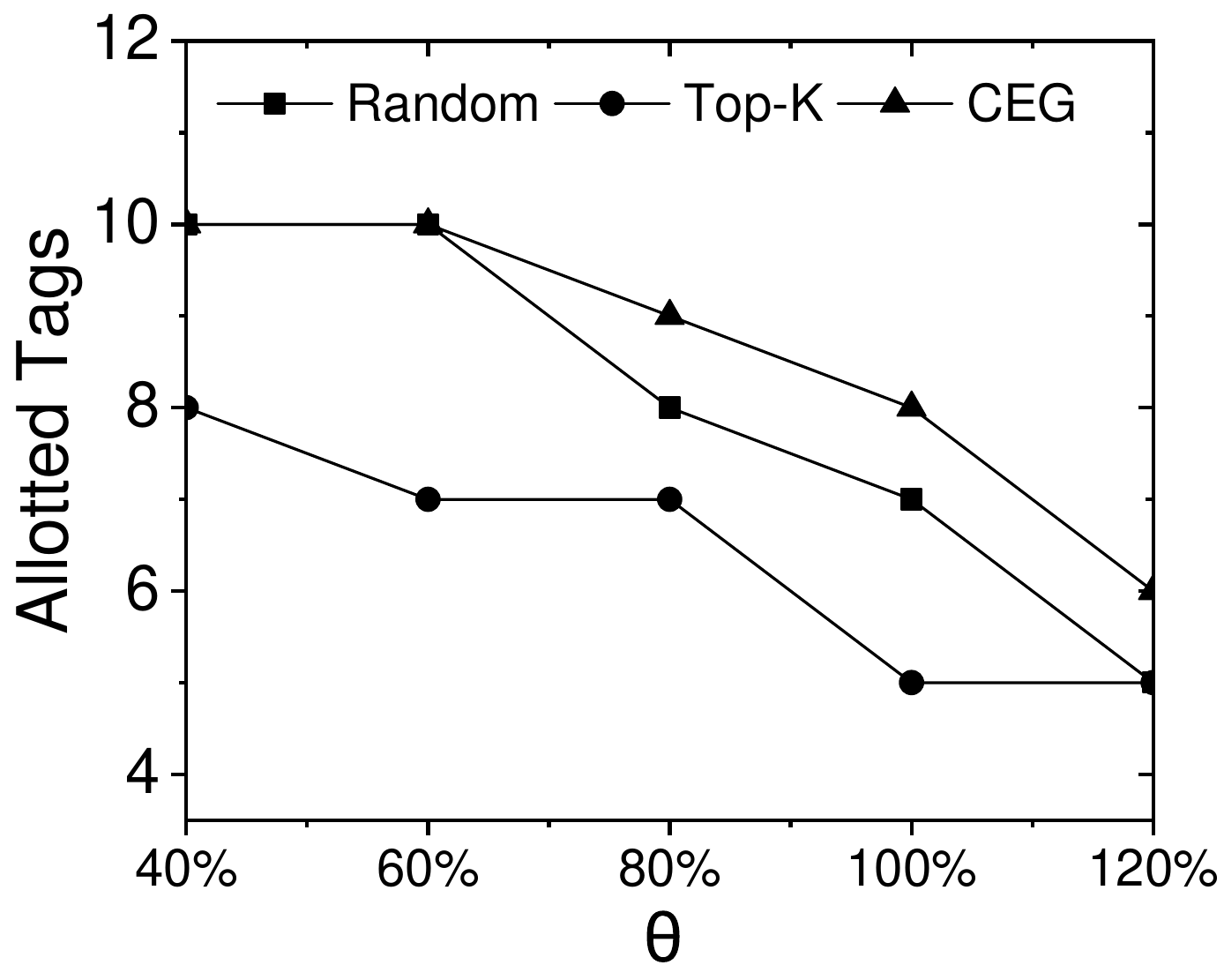}  \\
{\tiny (a) $\delta = 1\%, |\mathcal{T}| = 100$} &
{\tiny (b) $\delta = 2\%, |\mathcal{T}| = 50$} &
{\tiny (c) $\delta = 5\%, |\mathcal{T}| = 20$} &
{\tiny (d) $\delta = 10\%, |\mathcal{T}| = 10$} \\[5pt]
\includegraphics[scale=0.125]{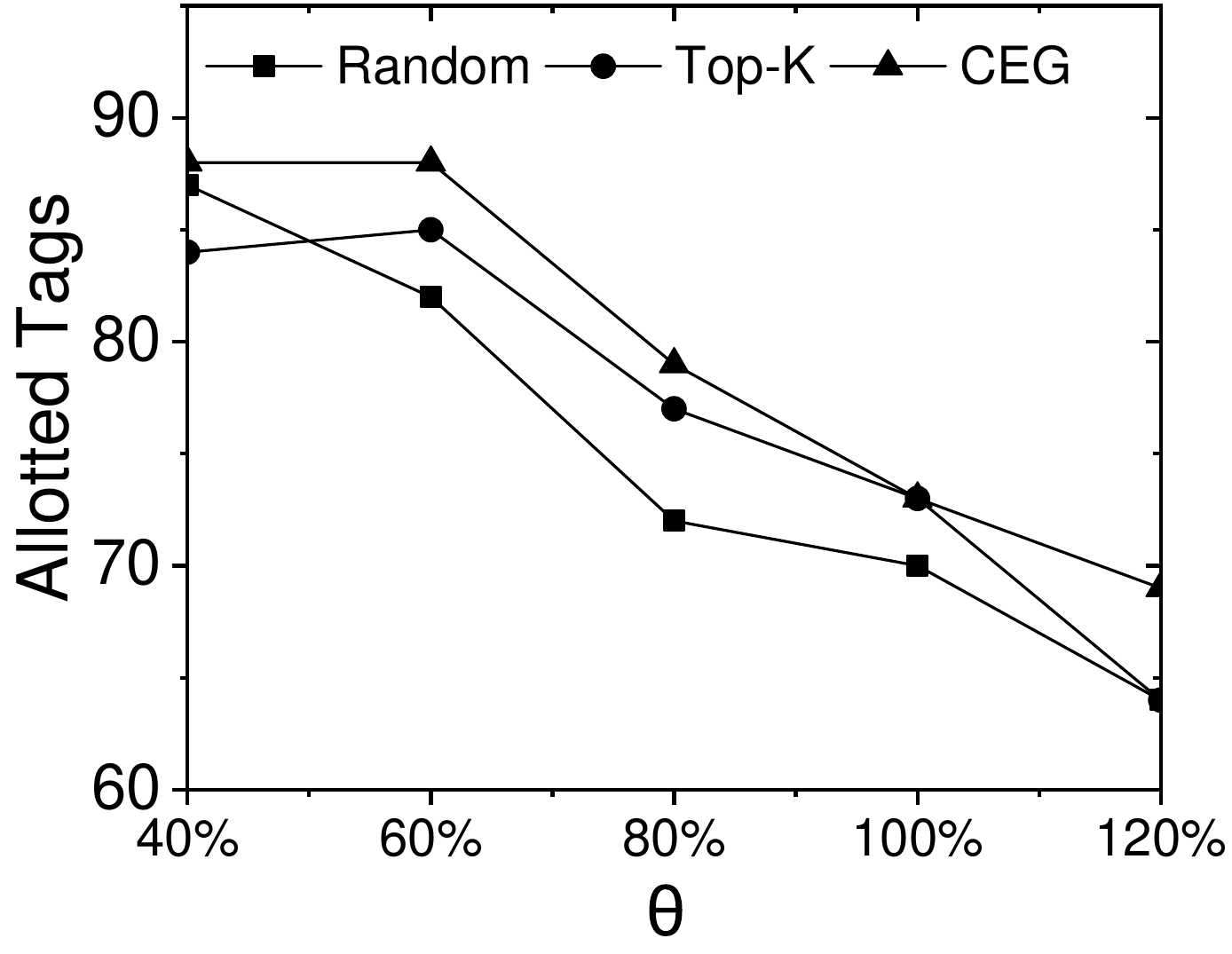} & \includegraphics[scale=0.13]{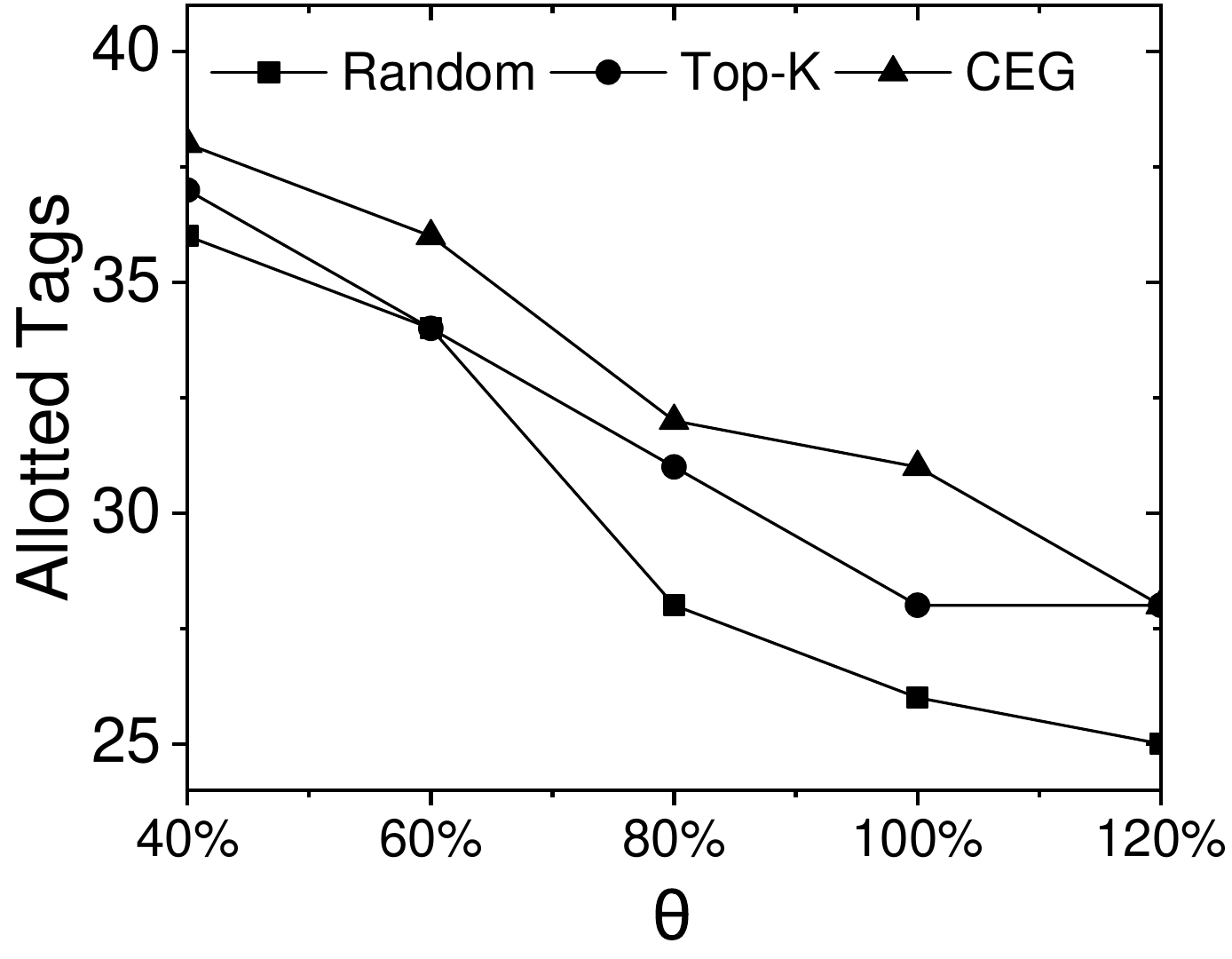} & \includegraphics[scale=0.125]{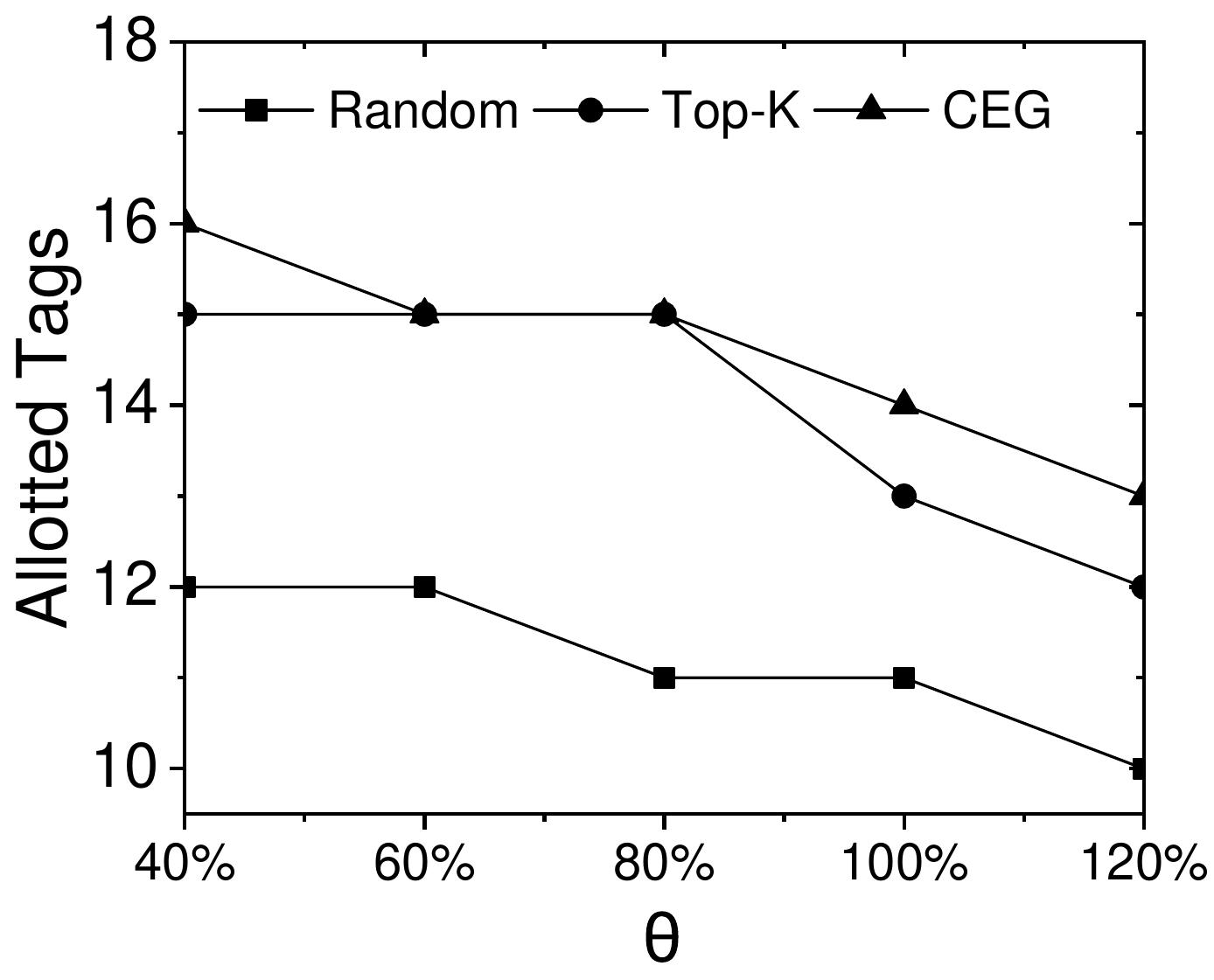} & \includegraphics[scale=0.125]{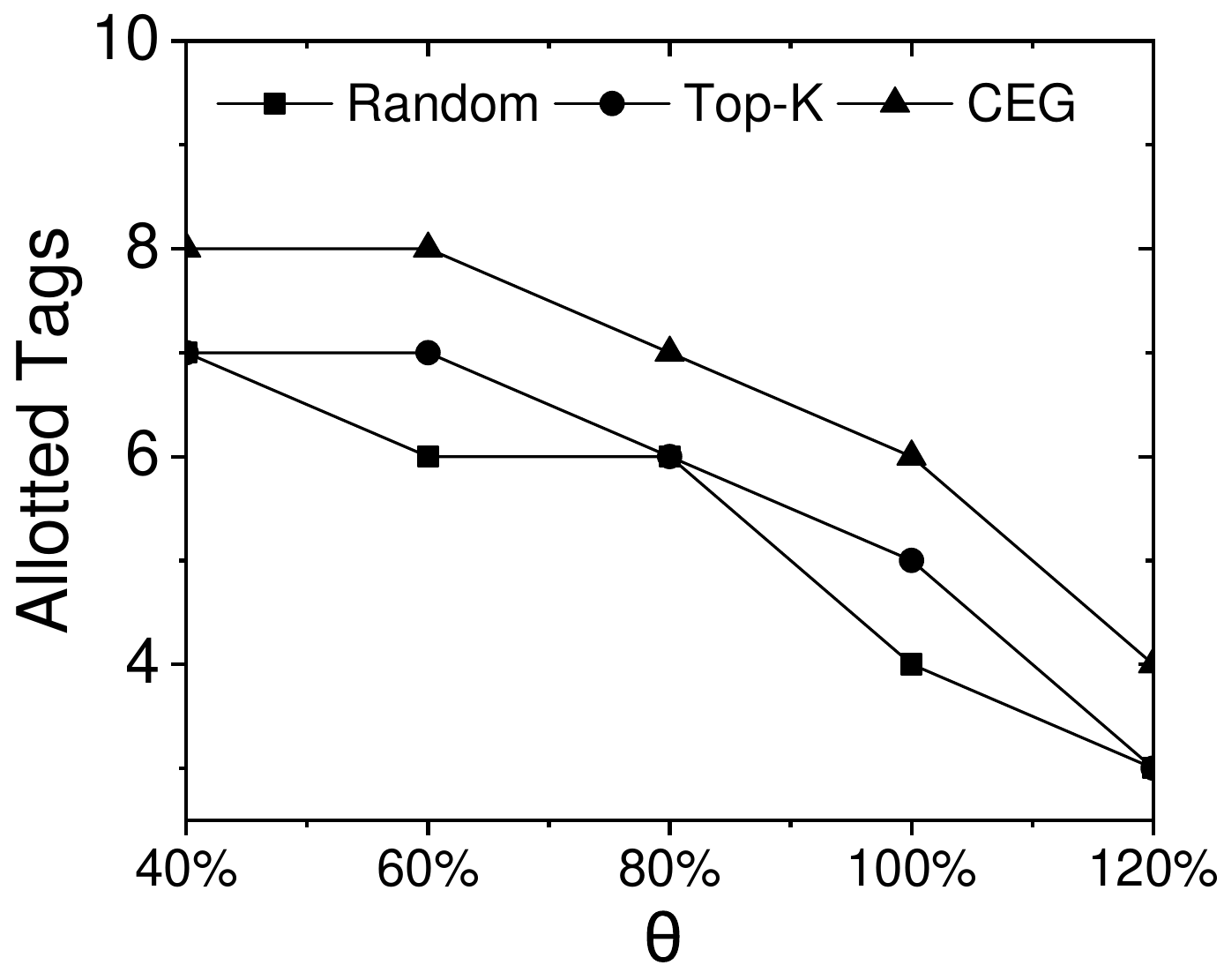}  \\
{\tiny (e) $\delta = 1\%, |\mathcal{T}| = 100$} &
{\tiny (f) $\delta = 2\%, |\mathcal{T}| = 50$} &
{\tiny (g) $\delta = 5\%, |\mathcal{T}| = 20$} &
{\tiny (h) $\delta = 10\%, |\mathcal{T}| = 10$} \\[5pt]
\end{tabular}
\caption{Allotted tags on varying $\theta$ in LA $(a,b,c,d)$, NYC $(e,f,g,h)$}
\label{Fig:ADV}
\end{figure*}

% TIME PLOT

\begin{figure*}[!ht]
\centering
\begin{tabular}{cccc}
\includegraphics[scale=0.125]{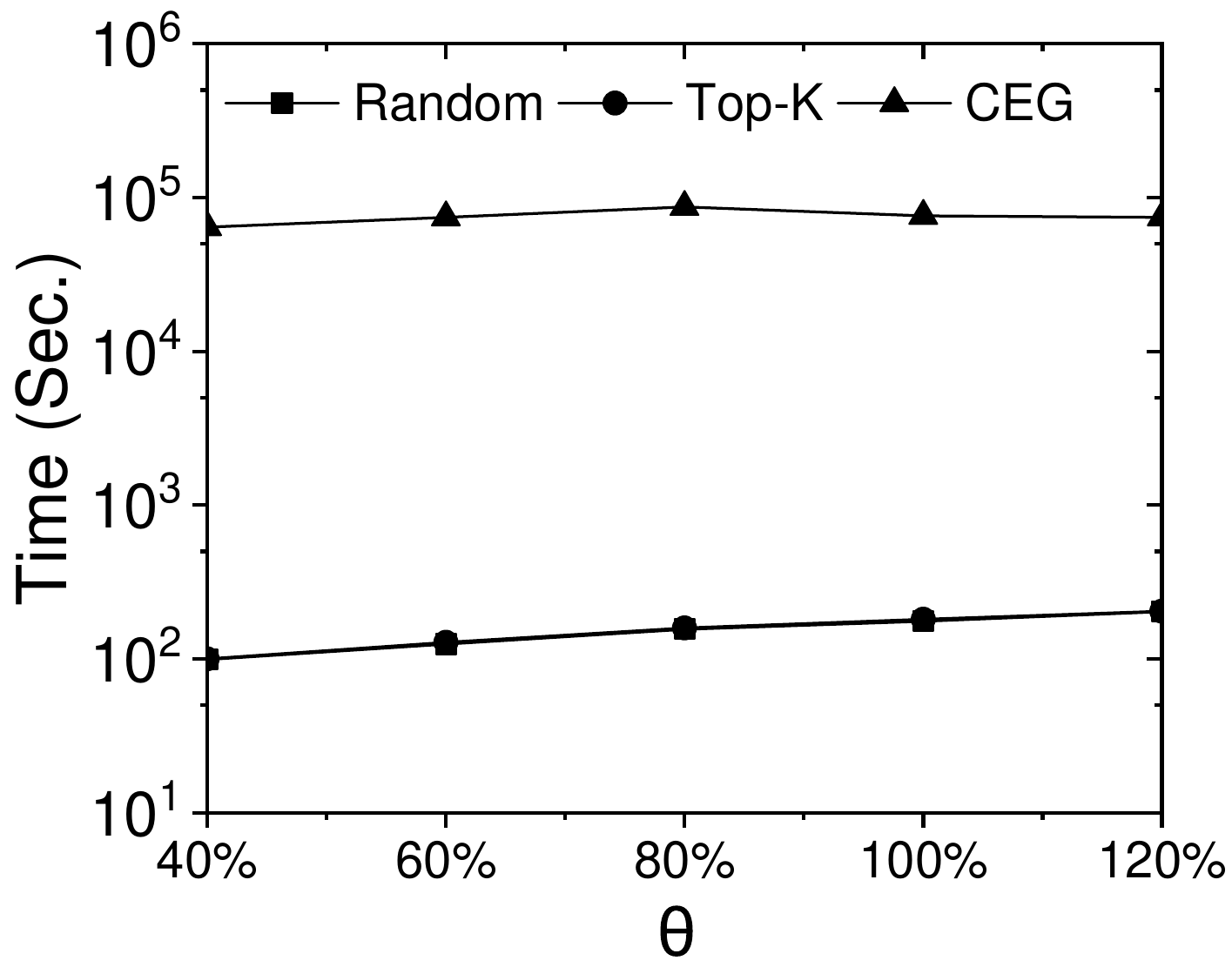} & \includegraphics[scale=0.125]{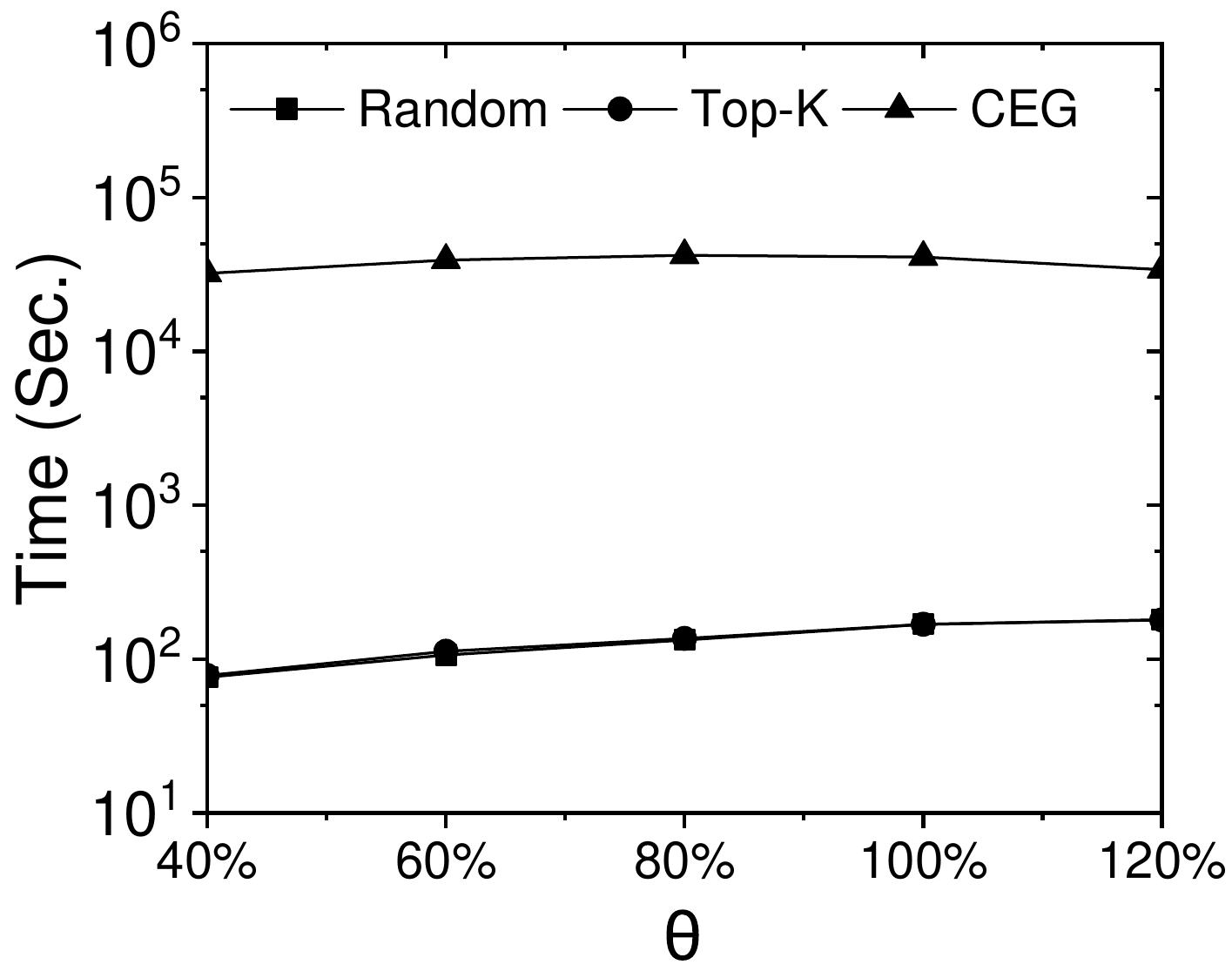} & \includegraphics[scale=0.125]{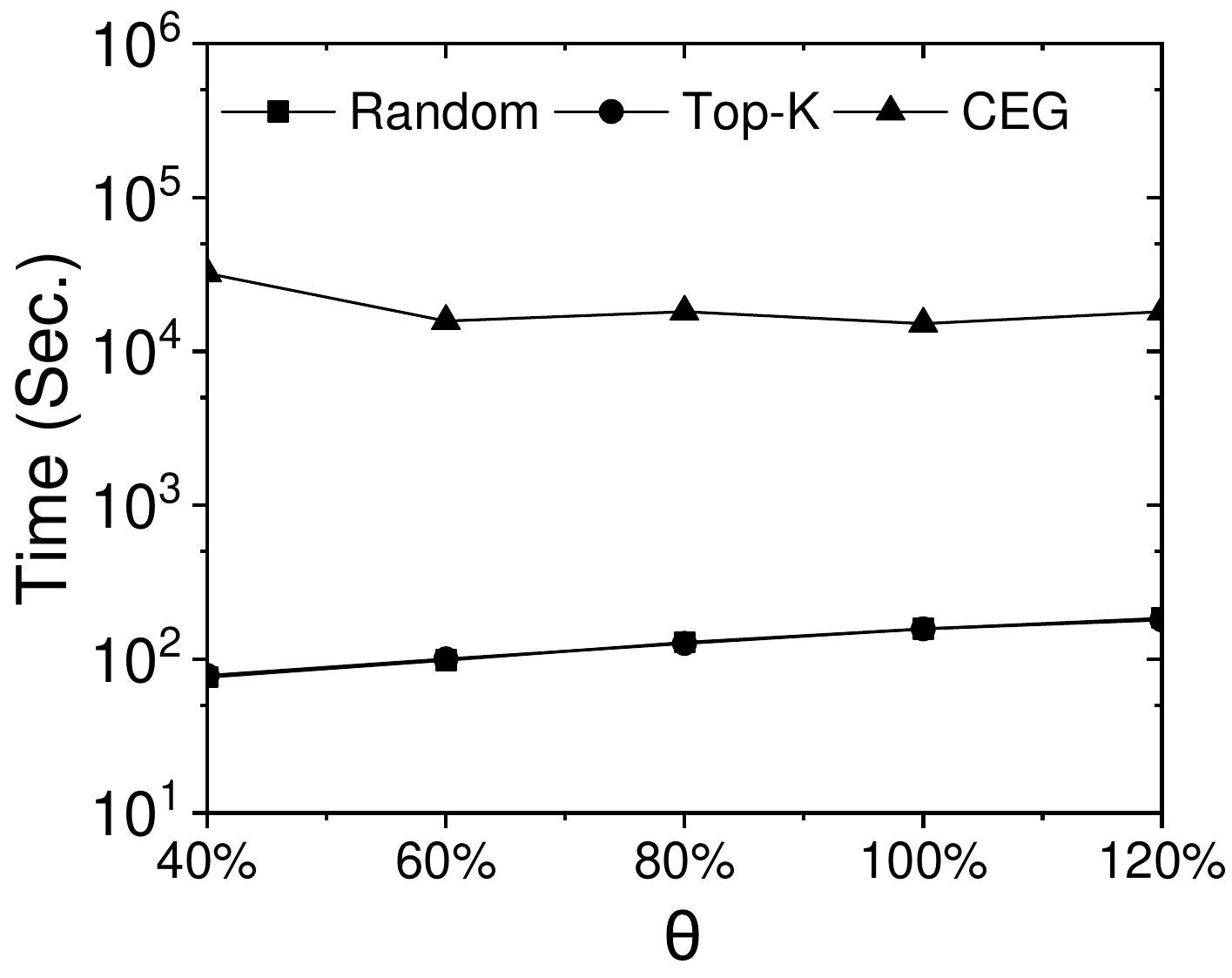} & \includegraphics[scale=0.125]{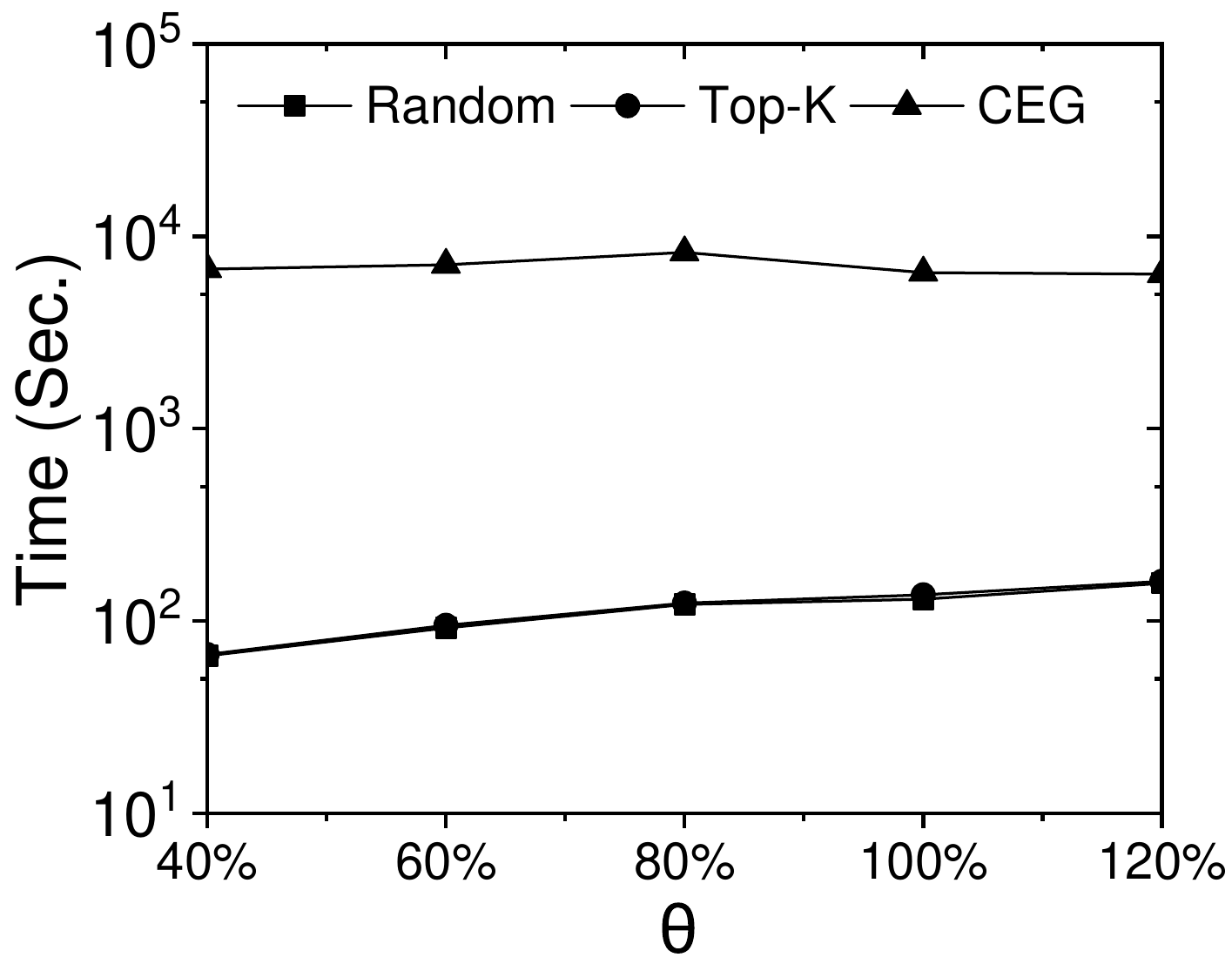}  \\
{\tiny (a) $\delta = 1\%, |\mathcal{T}| = 100$} &
{\tiny (b) $\delta = 2\%, |\mathcal{T}| = 50$} &
{\tiny (c) $\delta = 5\%, |\mathcal{T}| = 20$} &
{\tiny (d) $\delta = 10\%, |\mathcal{T}| = 10$} \\[5pt]
\includegraphics[scale=0.125]{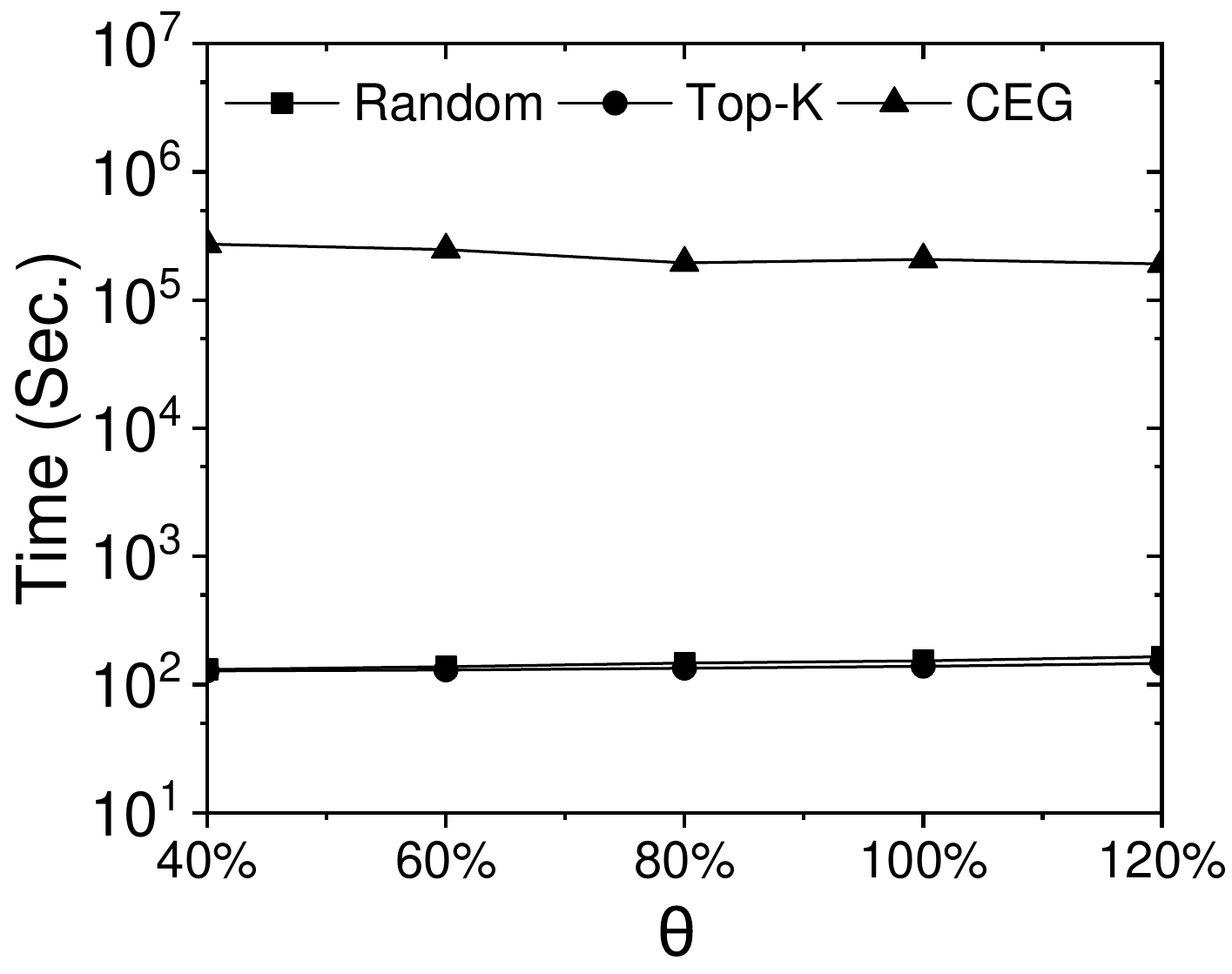} & \includegraphics[scale=0.125]{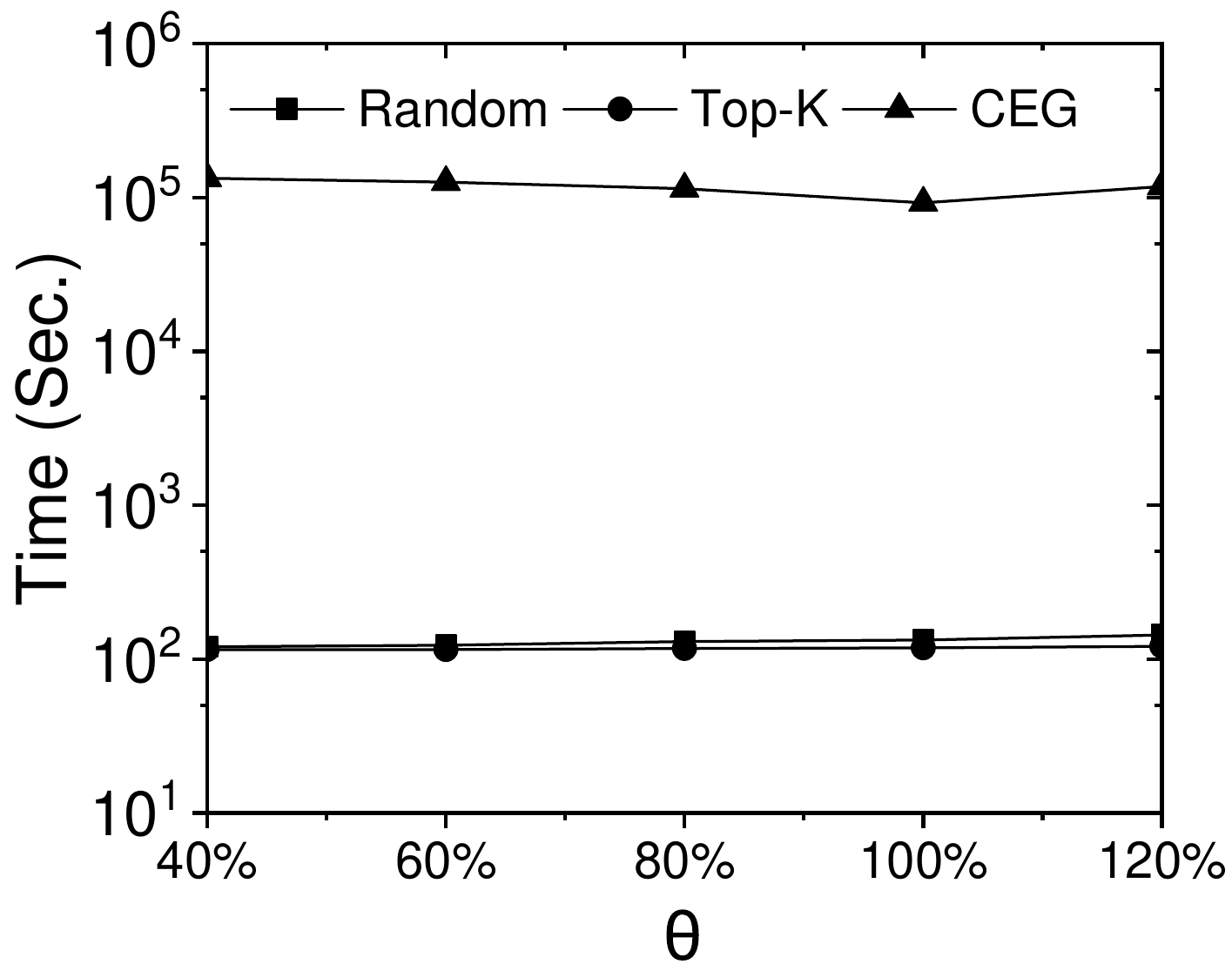} & \includegraphics[scale=0.125]{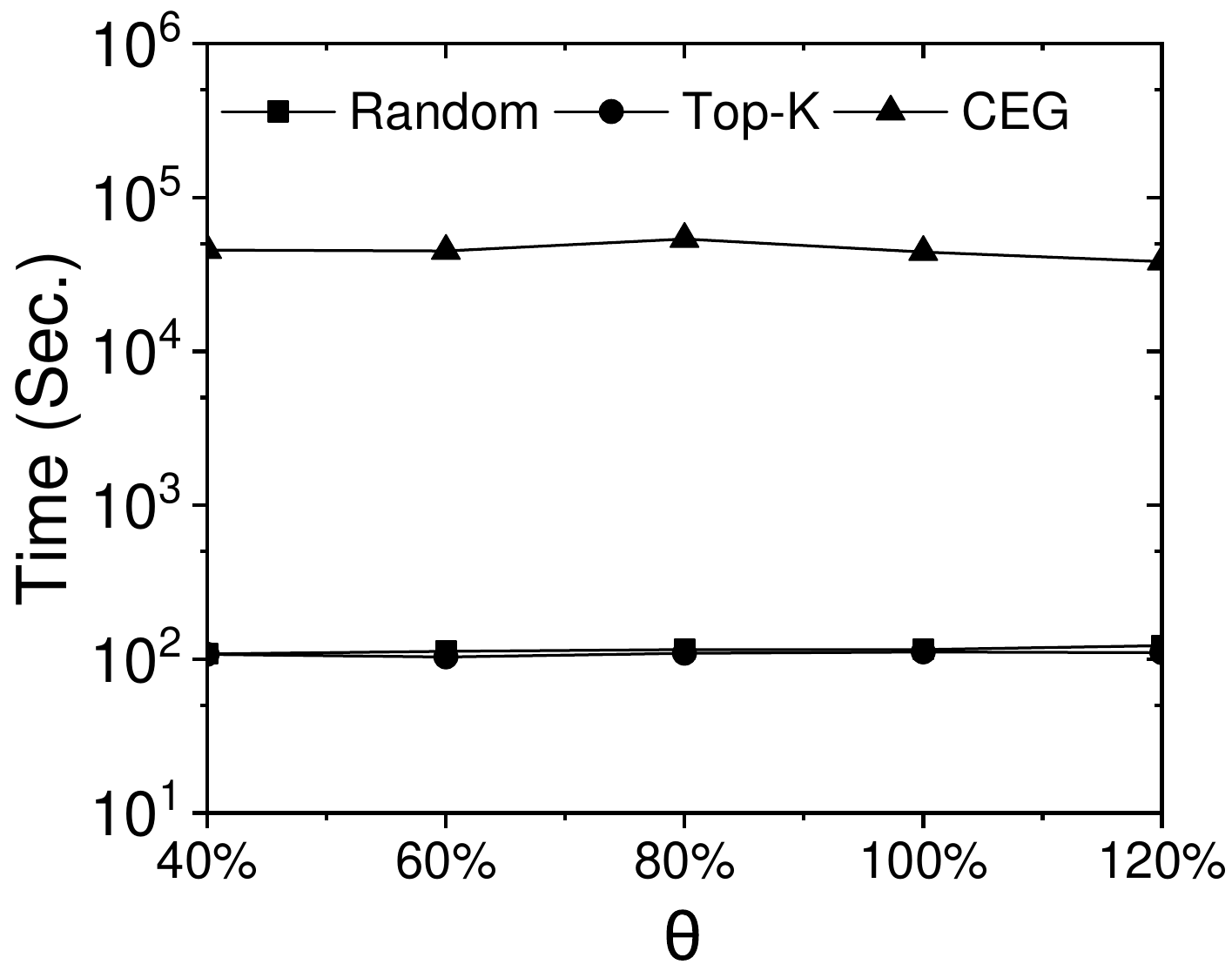} & \includegraphics[scale=0.125]{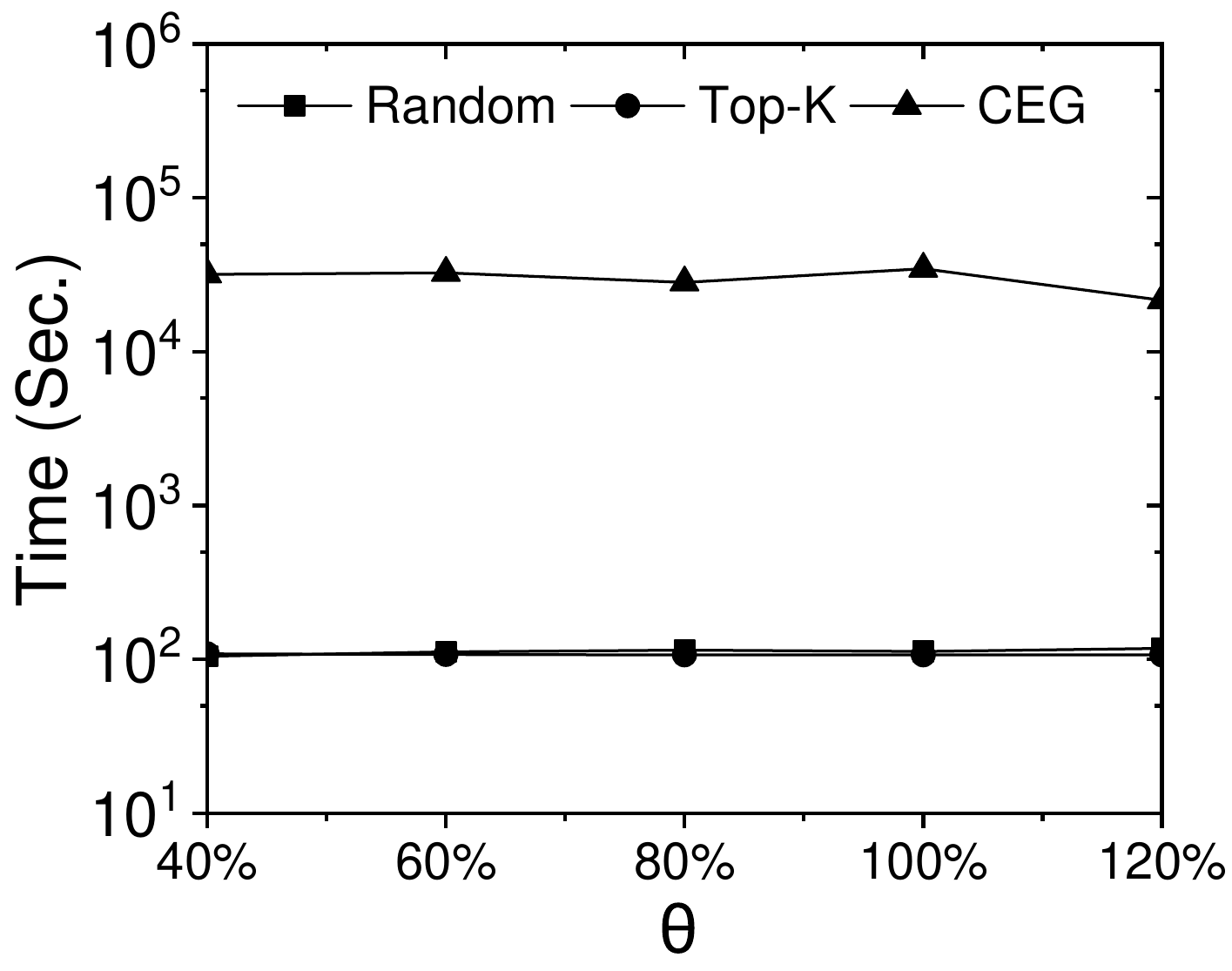}  \\
{\tiny (e) $\delta = 1\%, |\mathcal{T}| = 100$} &
{\tiny (f) $\delta = 2\%, |\mathcal{T}| = 50$} &
{\tiny (g) $\delta = 5\%, |\mathcal{T}| = 20$} &
{\tiny (h) $\delta = 10\%, |\mathcal{T}| = 10$} \\[5pt]
\end{tabular}
\caption{Efficiency study in LA $(a,b,c,d)$, NYC $(e,f,g,h)$}
\label{Fig:TIME}
\end{figure*}
% \vspace{-0.25in}
% COST PLOT
\begin{figure*}[!ht]
\centering
\begin{tabular}{cccc}
\includegraphics[scale=0.125]{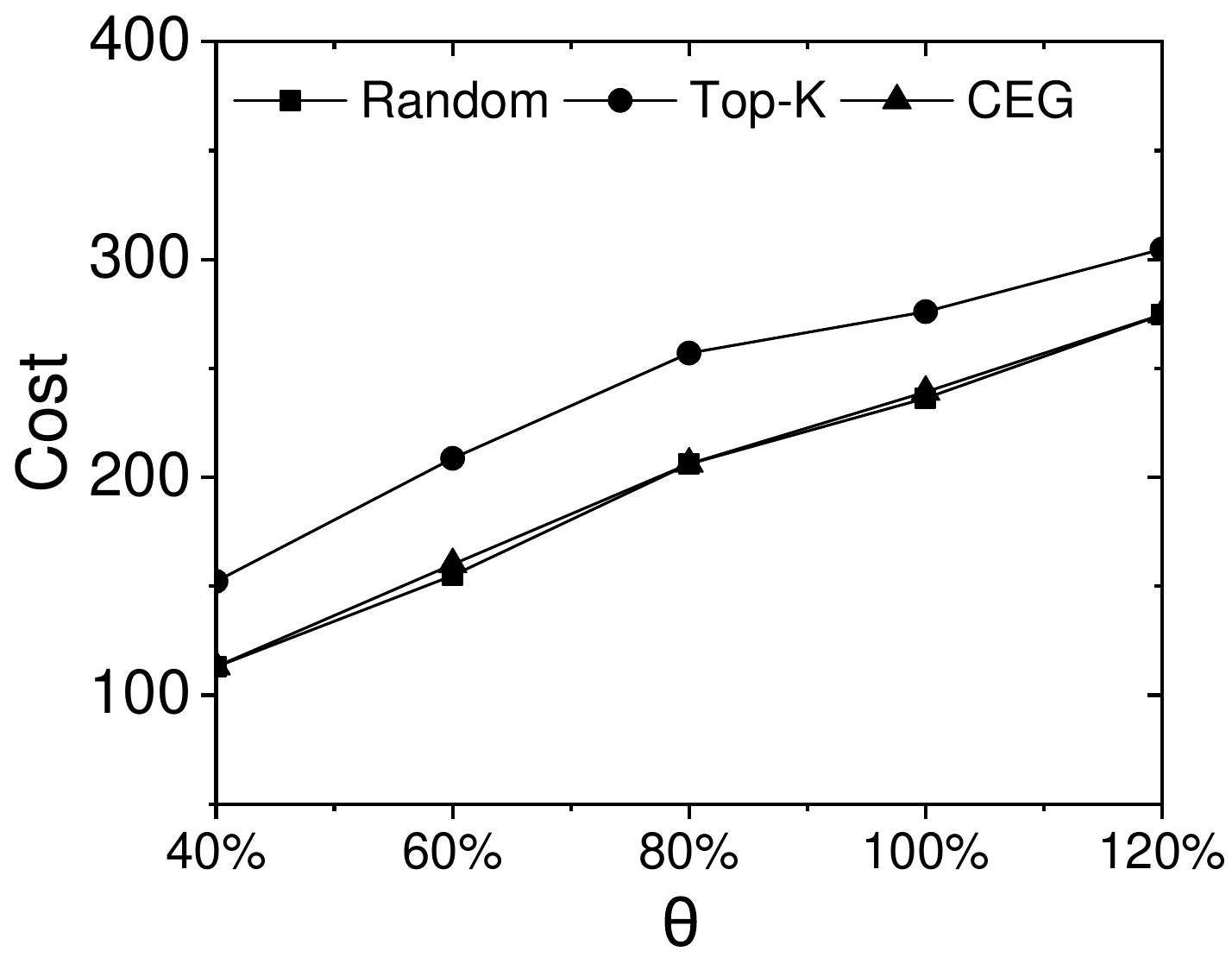} & \includegraphics[scale=0.125]{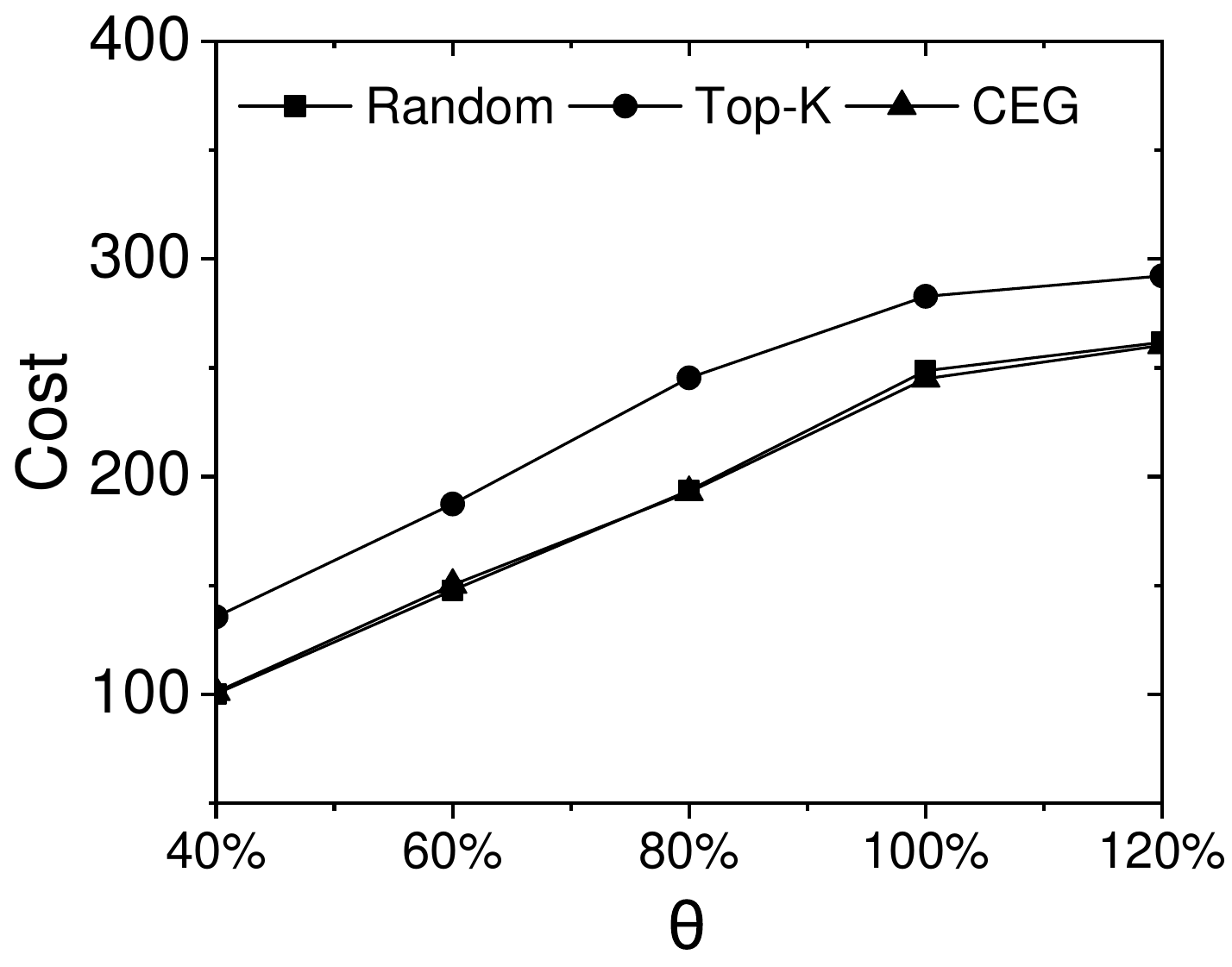} & \includegraphics[scale=0.125]{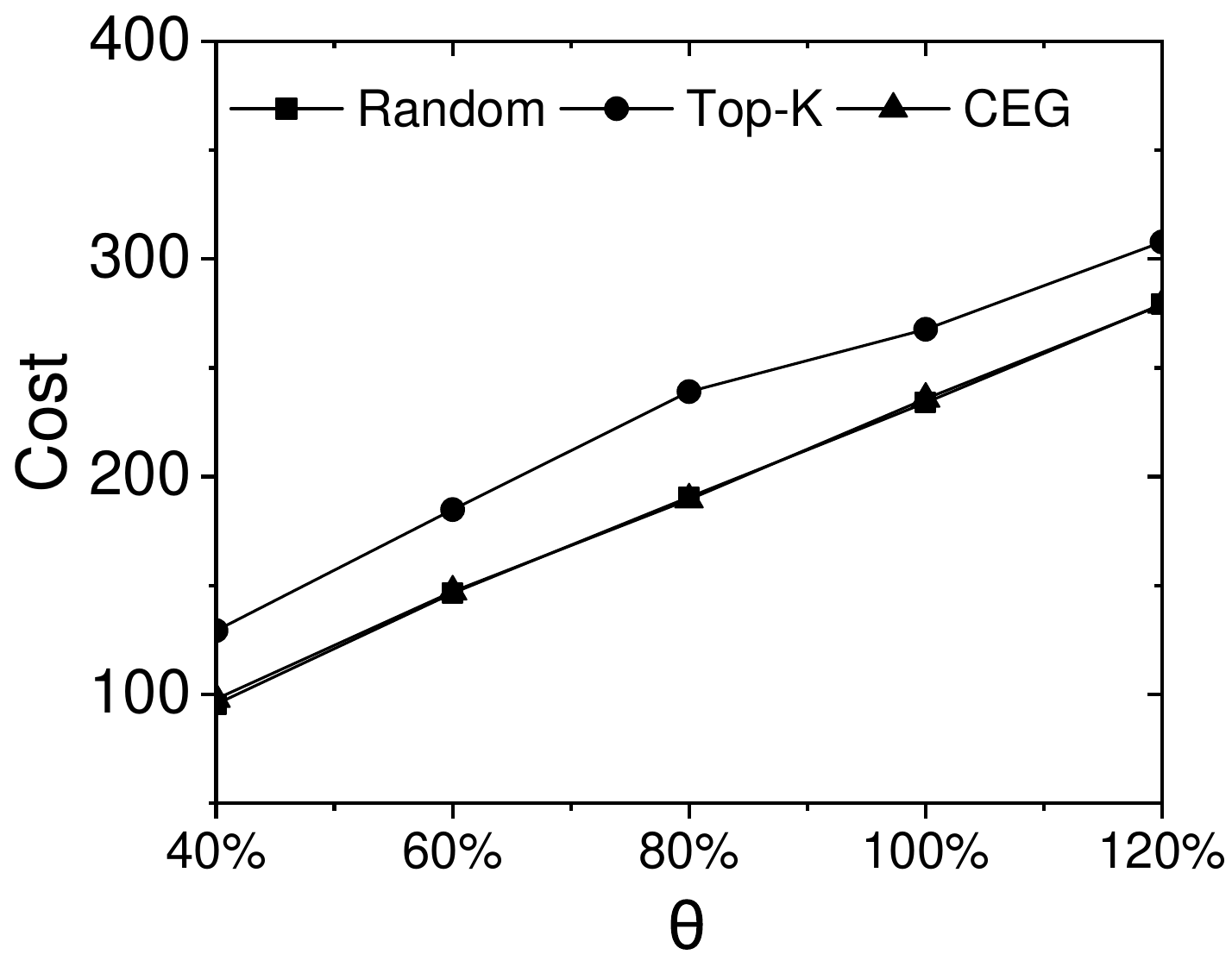} & \includegraphics[scale=0.125]{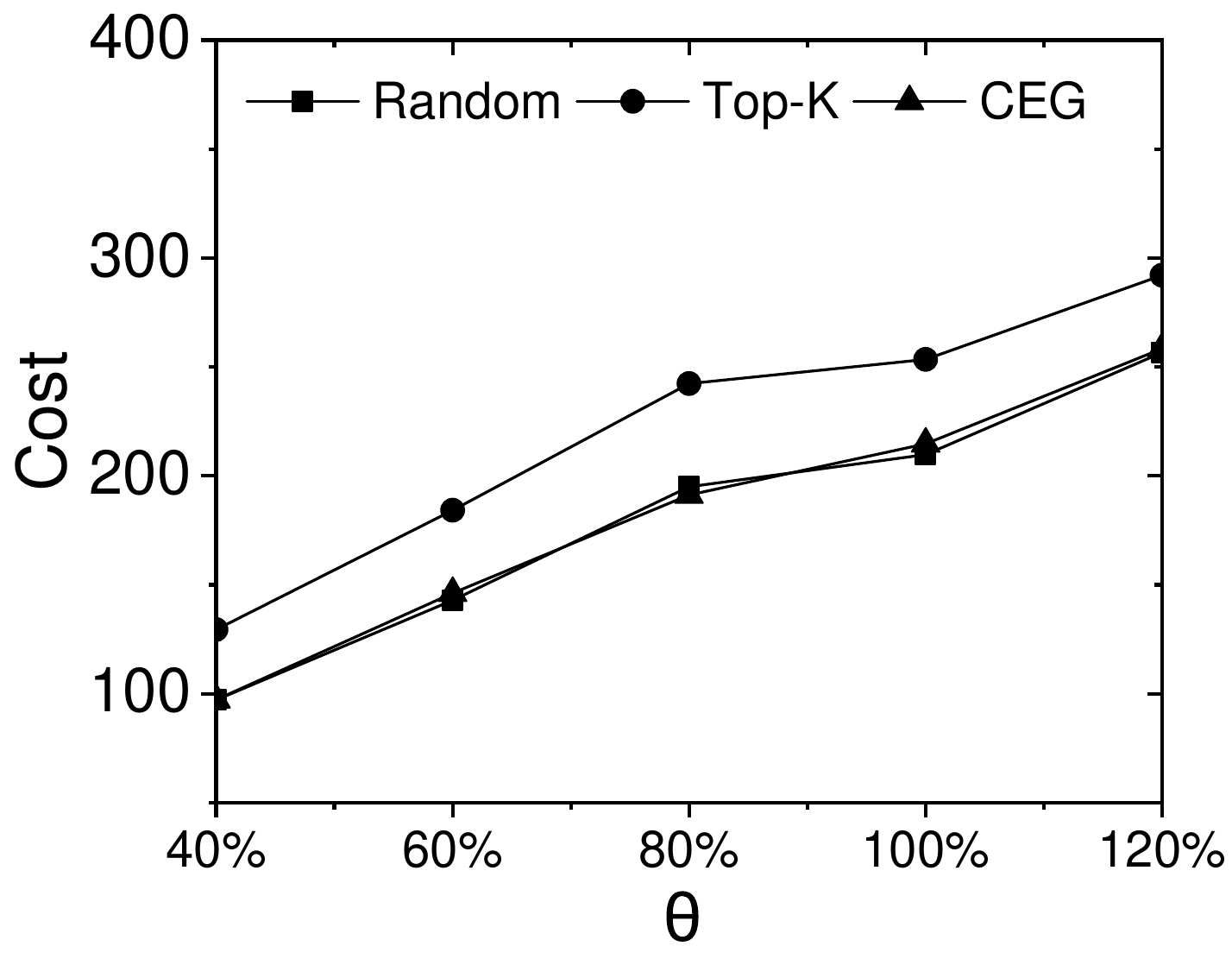} \\
{\tiny (a) $\delta = 1\%, |\mathcal{T}| = 100$} &
{\tiny (b) $\delta = 2\%, |\mathcal{T}| = 50$} &
{\tiny (c) $\delta = 5\%, |\mathcal{T}| = 20$} &
{\tiny (d) $\delta = 10\%, |\mathcal{T}| = 10$} \\[5pt]
\end{tabular}
\caption{ Utilized Cost on varying $\theta$ in LA $(a,b,c,f)$}
\label{Fig:COST}
\end{figure*}
% \vspace{-0.25in}

\paragraph{\textbf{ Varying $\theta$, $\delta$ Vs. Time.}}
The efficiency study of our work is important as, in real-life scenarios, an influence provider has thousands of billboard slots in a city like NYC and LA and may have new advertisers every day with a large number of tags for advertisement. In our study, we have three main observations from experiments, which are reported in Figure \ref{Fig:TIME}. First, in most cases, almost equal run time was achieved by both `Random' and `Top-$k$.' They incur much lower runtime compared to the `CEG' approach. This is because of the greedy selection of slots for each tag to be allocated. Second, with the increase of $\theta$ value, more slots need to be deployed to each tag, which requires more computational time. When $\theta \geq 100\%$, the individual demand, and global demand increase, and for this, the run time also increases. Third, with the increase of $\delta$, the run time decreases because the number of tags decreases.

\paragraph{\textbf{ Varying $\theta$, $\delta$ Vs. Cost.}}
The utilization of the budget shows the effectiveness of the proposed `CEG' approach. The findings of the NYC dataset are similar to the LA dataset. We have three main observations from our experiment, as shown in Figure \ref{Fig:COST}. First, with the increases of $\theta$, the cost of allocating tags to slots also increases. This happens because the zone-specific influence demand increases and the influence provider has to assign more slots. Second, the Top-$k$ takes more cost than the `Random' and `CEG' approach because Top-$k$ always picks the most influential slots, which have a high cost. On the other hand, `CEG' chooses slots that take minimum cost to allocate tags. Third, with the increase of $\delta$, the number of advertisers decreases, and allocation cost increases due to the higher influence demand of the tags.

\paragraph{\textbf{Additional Discussion.}}
The additional parameters used in our experiments are distance $(\lambda)$ and different constraints like zonal influence and budget constraints. First, we start describing $\lambda$. In our experiment, we vary $\lambda$ from $25m$ to $150m$ and observe that the influence value also increases with the $\lambda$ value increase. In our experiments, we set $\lambda$ as $100m$ as the default setting. Second, we observe that the increase in the number of zones leads to a smaller number of tags to be handled in the case of the NYC dataset compared to the LA dataset, as LA divided only three zones and NYC divided into five demographic zones. Third, as previously discussed, the computational cost also increases with the increase in the number of tags. Due to the space limitation, we are not able to show experimental observations for the setting $\delta = 20\%, |\mathcal{T}| = 5$ in all the effectiveness and efficiency studies. However, we have some observations of this setting. First, varying $\theta$ value number of handled tags are very less due to higher individual influence demand of the tags. Second, it requires more computational time as the influence provider tries to satisfy all the tags zone-specific high influence demand. Third, the allocation cost for the tag to slots increases due to higher global and individual influence demand.

\section{Concluding Remarks}\label{Sec:Conclusion}
In this paper, we have studied the tag assignment problem in billboard advertisements under zonal influence constraint. We have done an integer programming formulation for this problem. We have proposed a cost-effective greedy algorithm. A complexity analysis of the proposed solution approach has been conducted. We have also shown that the proposed solution approach produces $(\mathcal{P}^{*} +1)$-factor approximation ratio where $\mathcal{P}^{*} = \underset{j \in \mathcal{T}}{Max~~ \mathcal{P}_{j}}$. Experimental evaluation of real-world datasets shows the effectiveness and efficiency of the proposed solution approach. Now, our study can be extended in the following directions. First, this work can be extended to the multi-advertiser setting. Second, this problem can be extended to the influence provider perspective.

\begin{credits}
\subsubsection{\ackname} This work is supported by the Start-Up Research Grant provided by the Indian Institute of Technology Jammu, India (Grant No.: SG100047).
\end{credits}
\bibliographystyle{splncs04}
\bibliography{Paper}

\end{document}